%% file: main.tex
\documentclass[journal,onecolumn]{IEEEtran}

\usepackage[utf8]{inputenc}
\usepackage{amsmath,amssymb,amsfonts}
\usepackage{amsthm}
\usepackage{enumerate}
\usepackage{bm}
\usepackage{graphicx}
\usepackage{microtype}
\newtheorem{proposition}{Proposition}
\newtheorem{theorem}{Theorem}

\newtheorem{remark}{Remark}

\newtheorem{assumption}{Assumption}
\newtheorem{lemma}{Lemma}

\input{numbers.tex}

\title{Koopman Generator Decomposition for Port-Hamiltonian Systems}

\author{Victor M. Preciado,~\IEEEmembership{Senior Member,~IEEE}%
	\thanks{V. M. Preciado is with the Department of Electrical and Systems Engineering,
		University of Pennsylvania, Philadelphia, PA 19104 USA
		(e-mail: \texttt{preciado@seas.upenn.edu}).}}

\begin{document}
	\maketitle
	
	\begin{abstract}
        Port-Hamiltonian energy balance and Koopman approximation describe
        different aspects of nonlinear dynamics. This paper distinguishes
        exact observable identities, measure-dependent Galerkin projections,
        constrained surrogate identification, and certification of a
        controller on the nonlinear plant. Conservative weak skewness requires
        compatible measure and boundary conditions. Dissipative positivity
        belongs to the metric bracket; it does not generally make the
        projected first-order dissipative derivation positive. We give an
        explicit fixed-metric compatibility condition and a
        constant--Hamiltonian dictionary obstruction, and state the
        assumptions under which surrogate certificates imply plant bounds.
        A reproduced pendulum experiment illustrates the practical gap:
        its learned controller admits an exact shaped-energy certificate on
        the nonlinear plant, although its fitted quadratic storage is not a
        plant Lyapunov function and its archived residual-norm diagnostic is
        uninformative. These conclusions use the specified plant structure;
        they are not generic finite-data guarantees. The bracket, projection,
        and Lyapunov arguments are positioned relative to existing
        structured identification and residual-aware control methods.
    \end{abstract}


	\section{Introduction}
	\label{sec:intro}
	
	\emph{Port-Hamiltonian} (pH) systems encode energy structure
	through conservative interconnection, dissipation, and
	power-conjugate ports \cite{van2006port,schaft2014port}. For a
	Hamiltonian $\mathcal{H}$, input $\mathbf{u}$, output
	$\mathbf{y}$, and dissipation matrix $\mathbf{R}\succeq 0$, the
	resulting energy balance has the form
	\[
	\dot{\mathcal{H}}
	= -(\nabla\mathcal{H})^\intercal
	\mathbf{R}\nabla\mathcal{H}
	+ \mathbf{y}^\intercal\mathbf{u}
	\le \mathbf{y}^\intercal\mathbf{u}.
	\]
	This identity makes pH models useful for passivity-based design,
	including energy shaping~\cite{ortega2001energyshaping}, damping
	injection~\cite{ortega2001interconnection}, and modular
	interconnection of physical subsystems.
	
	The \emph{Koopman operator} approach represents nonlinear state
	dynamics through linear evolution of observables
	\cite{koopman1931hamiltonian,mezic2005spectral,mezic2013analysis,
		brunton2022modern}.
	Its infinitesimal generator differentiates observables along
	trajectories and therefore inherits structure from the vector
	field. For conservative Hamiltonian systems, the generator is
	skew with respect to an invariant measure under standard domain
	and boundary assumptions \cite{mezic2005spectral}. Dissipative pH
	systems do not generally preserve such measures, so conservative
	skewness and energy dissipation must be tracked separately at the
	generator level. This paper does so through a bracket
	decomposition that mirrors the pH structure of the vector field.
	
	The starting identity, recorded in
	Section~\ref{sec:decomposition}, is that for a smooth observable
	$f$ and an input $\mathbf{u}$, the pH dynamics induce the
	decomposition
	\[
	\mathcal{K}_{\mathbf{u}}f
	= \{f,\mathcal{H}\}_J - [f,\mathcal{H}]_R
	+ \nabla f^\intercal \mathbf{G}\mathbf{u},
	\]
	where $\mathcal{K}_{\mathbf{u}}$ is the directional derivative
	along the controlled pH vector field at fixed input $\mathbf{u}$,
	$\mathbf{J}$ is the skew-symmetric interconnection matrix,
	$\mathbf{R}$ is the symmetric positive semidefinite dissipation
	matrix, and $\mathbf{G}$ is the port map of the pH system. The
	interconnection bracket
	$\{f,g\}_J=\nabla f^\intercal \mathbf{J}\nabla g$ represents the
	conservative interconnection and the metric bracket
	$[f,g]_R=\nabla f^\intercal \mathbf{R}\nabla g$ represents the
	dissipative interconnection. The final term is the input-port
	derivation, that is, the first-order differential operator
	associated with the input vector field $\mathbf{G}\mathbf{u}$.
	
	The dissipative structure of a pH system enters the Koopman
	generator in a way that is worth stating precisely. The metric
	bracket evaluated on matching arguments is pointwise nonnegative,
	$[f,f]_R = \nabla f^\intercal \mathbf{R} \nabla f \ge 0$, and
	specializing this to the Hamiltonian observable yields the energy
	identity
	\[
	\mathcal{K}_{\mathbf{u}}\mathcal{H}
	= -[\mathcal{H},\mathcal{H}]_R
	+ \mathbf{y}^\intercal\mathbf{u}
	\le \mathbf{y}^\intercal\mathbf{u},
	\]
	from which the passivity inequality follows. The first-order derivation $f\mapsto [f,\mathcal{H}]_R$ is not, by contrast, a positive operator, that is, an operator $\mathcal{A}$ with $\langle f,\mathcal{A}f\rangle\ge0$ for all observables $f$. The pH structure therefore contributes positivity to the Koopman
	picture through the bracket, not through the derivation. This
	observation organizes the rest of the paper. The positive
    Dirichlet matrix and the Galerkin matrix of the dissipative derivation
    are distinct. A fitted pH surrogate imposes algebraic structure; the
    bracket identities alone do not imply that the unconstrained Galerkin
    drift has that structure in a prescribed storage metric.

    We distinguish four objects: the physical plant, a population weak
    projection, an empirically fitted surrogate, and the controller
    implemented on the plant. A sampling measure defines the projection
    objective; it need not be invariant under either the conservative or
    the dissipative dynamics. Conservative weak skewness requires the
    stated compatibility and boundary conditions. The Dirichlet matrix
    remains positive semidefinite for any nonnegative sampling measure
    for which its entries exist, including an exact empirical Gram sum.

    The control question is whether a controller designed from the
    surrogate can be certified on the original plant. Surrogate passivity
    alone cannot answer it. A residual-based route requires a uniform
    approximation bound, trajectory confinement, and state separation
    when a lifted-state conclusion is transferred to the physical state.
    Alternatively, a physical storage can be checked directly along the
    implemented closed-loop plant. The latter uses physical model
    information and does not follow from training data alone.

	The analysis has three components:
    \textbf{(i)} a condition-explicit account of observable brackets,
    mass-matrix Galerkin projection, and imposed surrogate structure,
    including a fixed-metric compatibility condition and a dictionary
    obstruction;
    \textbf{(ii)} a separation of standard surrogate passivity guarantees
    from conditional residual-based and direct physical plant
    certificates; and
    \textbf{(iii)} reproducible examples that expose the distinction
    between fitted-storage decay, residual magnitude, and the stability
    of the implemented controller. The elementary bracket decomposition,
    PSD projection, and Lyapunov/LaSalle mechanisms are established tools,
    not new identification or general data-only certification methods.

	Koopman analysis of conservative Hamiltonian systems is treated
	in \cite{mezic2005spectral}. Koopman identification with stability
	constraints and system-norm regularization is developed using
	matrix-inequality formulations and sequential convex
	optimization in~\cite{dahdah2022systemnorm}.
	Structure-preserving identification has also been pursued through
	symplectic learning~\cite{yildiz2024}, and data-driven
	port-Hamiltonian identification from snapshots through dynamic
	mode decomposition and non-intrusive
	realization~\cite{morandin2023,cherifi2022}. Related work learns
	Koopman models under prescribed quadratic dissipativity
	constraints~\cite{hara2020dissipativity}. These approaches impose
	desired stability, dissipativity, or geometric properties during
	identification.
	The present analysis separates these imposed constraints from exact
    weak-generator identities. In particular, the fitted dissipation
    matrix is not automatically the Dirichlet matrix. Weighted pH-DMD
    already supplies a spectral decomposition of the structural fitting
    loss and discusses its power-balance interpretation
    \cite[Theorem~3.7 and Remark~3.8]{morandin2023}; spectral clipping or a
    structural-loss penalty alone is therefore not a novelty claim here.

    Residual-aware Lyapunov arguments already connect stable learned
    Koopman models to nonlinear dynamics under proportional-error
    conditions \cite[Theorem~3]{mamakoukas2020stable}. SafEDMD develops
    proportional-error bounds and controller guarantees under its own
    approximation and sampling assumptions \cite{straesser2026safedmd}.
    Neural Koopman identification also explicitly addresses transfer of
    dissipativity from learned models to the nonlinear system under
    approximation-error assumptions \cite{xu2025dissipativity}.
    These works preclude a priority claim for the general transfer question.
    Our pendulum analysis is a verification of a particular learned
    controller using the specified plant equations, rather than a
    replacement for those learning guarantees.

    Skew and symmetric brackets also occur in metriplectic and GENERIC
    formulations \cite{kaufman1984,morrison1986,grmela1997,grmela1997b,
    ottinger2005,morrison2024}, and double-bracket dissipation has a
    related geometric setting \cite{bloch1996}. Standard thermodynamic
    formulations additionally distinguish energy and entropy generators
    and impose degeneracy conditions. Generic pH damping decreases the
    represented Hamiltonian, so the present construction is a related
    bracket form, not a claimed specialization satisfying those
    thermodynamic conditions. The operator setting follows the standard
    systems-and-control framework \cite{mauroy2020}.

	The paper is organized as follows. Section~\ref{sec:prelim}
	reviews pH and Koopman preliminaries, assumptions, and notation.
	Section~\ref{sec:decomposition} develops the bracket
	decomposition. Section~\ref{sec:approximation} develops the
	finite-dimensional matrix representation of the generator and
	constructs the lifted port-Hamiltonian surrogate.
	Section~\ref{sec:control} treats passivity-based control for the
	surrogate. Section~\ref{sec:example} presents examples, and
	Section~\ref{sec:conclusion} concludes.

	\section{Preliminaries}
	\label{sec:prelim}
	
	\subsection{Port-Hamiltonian Systems}
	
	Let $\mathcal{X}\subseteq\mathbb{R}^n$ be the state domain, with
	state $\mathbf{x}\in\mathcal{X}$, input
	$\mathbf{u}\in\mathbb{R}^m$, and output
	$\mathbf{y}\in\mathbb{R}^m$. A nonlinear port-Hamiltonian (pH)
	system has the form~\cite{van2006port,schaft2014port}
	\begin{equation}
		\begin{aligned}
			\dot{\mathbf{x}}
			&= \bigl(\mathbf{J}(\mathbf{x}) - \mathbf{R}(\mathbf{x})\bigr)
			\nabla\mathcal{H}(\mathbf{x}) + \mathbf{G}(\mathbf{x})\,\mathbf{u},\\
			\mathbf{y}
			&= \mathbf{G}(\mathbf{x})^\intercal\nabla\mathcal{H}(\mathbf{x}),
		\end{aligned}
		\label{eq:ph_nonlinear}
	\end{equation}
	where $\mathcal{H}:\mathcal{X}\to\mathbb{R}$ is the Hamiltonian,
	$\mathbf{J}(\mathbf{x})$ is the interconnection matrix,
	$\mathbf{R}(\mathbf{x})$ is the dissipation matrix, and
	$\mathbf{G}(\mathbf{x})$ is the port map. The interconnection
	and dissipation matrices satisfy
	$\mathbf{J}(\mathbf{x}) = -\mathbf{J}(\mathbf{x})^\intercal$ and
	$\mathbf{R}(\mathbf{x}) = \mathbf{R}(\mathbf{x})^\intercal
	\succeq 0$. Along classical trajectories, the Hamiltonian
	satisfies the pointwise energy balance
	\begin{equation}
		\dot{\mathcal{H}}
		= \mathbf{y}^\intercal\mathbf{u}
		- (\nabla\mathcal{H})^\intercal\mathbf{R}(\mathbf{x})\,\nabla\mathcal{H}
		\;\le\; \mathbf{y}^\intercal\mathbf{u},
		\label{eq:ph_energy_balance}
	\end{equation}
	which gives passivity when $\mathcal{H}$ is bounded below on the
    considered state domain, so that an additive constant makes it a
    nonnegative storage~\cite{willems1972dissipative,van2006port}.
    The differential identity itself does not require this lower bound.
	
	Two algebraic consequences of these symmetry conditions will
	carry over to the Koopman generator in
	Section~\ref{sec:decomposition}. The skewness of $\mathbf{J}$
	gives $\xi^\intercal\mathbf{J}(\mathbf{x})\xi=0$ for all
	$\xi\in\mathbb{R}^n$, which underlies the conservative bracket
	identity. The positive semidefiniteness of $\mathbf{R}$ gives
	$(\nabla\mathcal{H})^\intercal\mathbf{R}\nabla\mathcal{H}\ge 0$,
	which underlies the dissipative bracket nonnegativity.
	
	For later use, we decompose the pH vector field
	$\mathbf{v}(\mathbf{x}, \mathbf{u}) = (\mathbf{J}(\mathbf{x}) -
	\mathbf{R}(\mathbf{x}))\nabla\mathcal{H}(\mathbf{x}) +
	\mathbf{G}(\mathbf{x})\mathbf{u}$ into its conservative,
	dissipative, and input components, together with the unforced
	drift:
	\begin{equation}
		\begin{aligned}
			\mathbf{v}_J(\mathbf{x})
			&\triangleq \mathbf{J}(\mathbf{x})\nabla\mathcal{H}(\mathbf{x}),\\
			\mathbf{v}_R(\mathbf{x})
			&\triangleq \mathbf{R}(\mathbf{x})\nabla\mathcal{H}(\mathbf{x}),\\
			\mathbf{v}_0(\mathbf{x})
			&\triangleq \mathbf{v}_J(\mathbf{x})-\mathbf{v}_R(\mathbf{x}),\\
			\mathbf{v}_u(\mathbf{x},\mathbf{u})
			&\triangleq \mathbf{G}(\mathbf{x})\mathbf{u}.
		\end{aligned}
		\label{eq:ph_vector_fields}
	\end{equation}
	The vector field $\mathbf{v}_0$ is the unforced drift of the pH
	system, and the bracket derivations of
	Section~\ref{sec:decomposition} are built from these fields.

	\begin{assumption}[Pointwise regularity]
		\label{ass:pointwise}
		The Hamiltonian satisfies $\mathcal{H}\in C^2(\mathcal{X})$,
		and $\mathbf{J}$, $\mathbf{R}$, and $\mathbf{G}$ are at least
		$C^1$ on $\mathcal{X}$. Pointwise generator identities are stated for observables $f\in C^1(\mathcal{X})$. Wherever the weighted divergence $\operatorname{div}_\mu$ appears, the density $\rho$ of $\mu=\rho\,d\mathbf{x}$ is $C^1$ and strictly positive.
	\end{assumption}
	
	In the proof of Proposition~\ref{thm:skewness}, a volume integral of
	the form $\int_{\mathcal{X}} \nabla(fg)^\intercal \mathbf{v} \,
	d\mu$ is converted into a divergence-theorem statement via
	integration by parts. The boundary term produced by this operation must
	vanish for the skewness identity to hold. The following
	assumption isolates this requirement.
	
	\begin{assumption}[Test function space]
		\label{ass:weak}
		The test space $\mathcal{V}$ consists of smooth real-valued
		functions on $\mathcal{X}$ such that, for any $f, g \in
		\mathcal{V}$, the integration-by-parts formula
		\[
		\int_{\mathcal{X}} \nabla(fg)^\intercal \mathbf{v} \, d\mu
		\;=\;
		-\int_{\mathcal{X}} fg \, \operatorname{div}_\mu \mathbf{v}
		\, d\mu
		\]
		holds for $\mathbf{v} \in \{\mathbf{v}_J, \mathbf{v}_R,
		\mathbf{v}_u\}$. Standard choices for $\mathcal{V}$ include
		$C_c^\infty(\mathcal{X})$ when $\mathcal{X}$ has nontrivial
		boundary, $C^\infty$ periodic functions when $\mathcal{X}$ is
		a torus, and $C_c^\infty(\mathbb{R}^n)$ when $\mathcal{X} =
		\mathbb{R}^n$.
	\end{assumption}
	
	\begin{remark}
		The assumption holds automatically for the three standard
		choices listed above. For noncompactly supported dictionaries on unbounded domains,
        integrability and vanishing flux at infinity must also be checked.
        On bounded domains a cutoff can enforce vanishing boundary flux,
        but it changes the dictionary and can destroy exact constant or
        Hamiltonian inclusion.
	\end{remark}
	
\subsection{Brackets Associated with the pH Structure}
\label{sec:brackets}

The bracket decomposition in
Section~\ref{sec:decomposition} is built from two bilinear forms
on smooth observables. For $f, g \in C^1(\mathcal{X})$, define
the \emph{interconnection bracket}
\begin{equation}
	\{f,g\}_J
	\triangleq \nabla f^\intercal \mathbf{J}\nabla g
	\label{eq:bracket_J}
\end{equation}
and the \emph{metric bracket}
\begin{equation}
	[f,g]_R
	\triangleq \nabla f^\intercal \mathbf{R}\nabla g.
	\label{eq:bracket_R}
\end{equation}
Skew-symmetry of $\mathbf{J}$ gives
$\{f,g\}_J = -\{g,f\}_J$, and positive semidefiniteness of
$\mathbf{R}$ gives
\begin{equation}
	[f,f]_R
	= \nabla f^\intercal\mathbf{R}\nabla f \ge 0
	\label{eq:metric_bracket_psd}
\end{equation}
pointwise. We call $\{\cdot,\cdot\}_J$ an interconnection bracket rather
than a Poisson bracket because the Jacobi identity
$\{f,\{g,h\}_J\}_J + \{g,\{h,f\}_J\}_J + \{h,\{f,g\}_J\}_J = 0$,
which holds for the canonical Poisson bracket of Hamiltonian
mechanics, may fail when $\mathbf{J}$ depends on the state. Only
skew-symmetry is used in what follows.

Applying the Koopman generator $\mathcal{K}f = \nabla f^\intercal
\mathbf{v}$ to each component of the pH vector-field
decomposition~\eqref{eq:ph_vector_fields} yields four
derivations:
\begin{equation}
	\begin{aligned}
		\mathcal{K}_J f
		&\triangleq \nabla f^\intercal \mathbf{v}_J = \{f,\mathcal{H}\}_J,\\
		\mathcal{K}_R f
		&\triangleq \nabla f^\intercal \mathbf{v}_R = [f,\mathcal{H}]_R,\\
		\mathcal{K}_0 f
		&\triangleq \nabla f^\intercal \mathbf{v}_0
		= \mathcal{K}_J f - \mathcal{K}_R f,\\
		\mathcal{K}_{\rm port}(\mathbf{u})f
		&\triangleq \nabla f^\intercal \mathbf{v}_u
		= \nabla f^\intercal \mathbf{G}\mathbf{u}.
	\end{aligned}
	\label{eq:bracket_derivations}
\end{equation}
The first equality in each line is the Koopman generator along
the corresponding vector field; the second equality rewrites it
using the bracket
notation~\eqref{eq:bracket_J}--\eqref{eq:bracket_R}, when
applicable. Proposition~\ref{prop:decomposition} shows that these four
derivations are not just a notational convenience: the full
Koopman generator decomposes as $\mathcal{K}_{\mathbf{u}} =
\mathcal{K}_J - \mathcal{K}_R + \mathcal{K}_{\rm port}(\mathbf{u})$,
and the Hamiltonian observable satisfies the energy identity
\[
\mathcal{K}_{\mathbf{u}}\mathcal{H}
= -[\mathcal{H},\mathcal{H}]_R + \mathbf{y}^\intercal\mathbf{u}
\le \mathbf{y}^\intercal\mathbf{u}.
\]

The nonnegativity in~\eqref{eq:metric_bracket_psd} is a property
of the bracket evaluated on matching arguments. It should not be
confused with positivity of the first-order derivation
$\mathcal{K}_R$, which does not hold on arbitrary observables
(see the counterexample in Section~\ref{sec:decomposition}).
	
	\subsection{Koopman Operator and Generator}
	\label{sec:koopman}
	
	Let $\varphi_t$ denote the flow of a smooth autonomous vector
	field $\dot{\mathbf{x}}=\mathbf{v}(\mathbf{x})$, that is, the
	map taking each initial state $\mathbf{x}$ to the trajectory
	point at time $t$. The \emph{Koopman operator}
	$\mathcal{U}^t$ at time $t$ is defined on smooth observables by
	composition with the flow,
	\begin{equation}
		(\mathcal{U}^t f)(\mathbf{x}) = f(\varphi_t(\mathbf{x})).
		\label{eq:koopman_op}
	\end{equation}
	Its \emph{infinitesimal generator} $\mathcal{K}$ is the time
	derivative of $\mathcal{U}^t$ at $t=0$, and coincides with the
	directional derivative of $f$ along the vector field
	$\mathbf{v}$:
	\begin{equation}
		\mathcal{K}f(\mathbf{x})
		= \lim_{t\to 0^+}
		\frac{(\mathcal{U}^t f)(\mathbf{x}) - f(\mathbf{x})}{t}
		= \nabla f(\mathbf{x})^\intercal \mathbf{v}(\mathbf{x}),
		\label{eq:koopman_gen}
	\end{equation}
	for $f\in C^1(\mathcal{X})$.
	
	The results below are pointwise identities for smooth
	observables, together with integration-by-parts identities on
	the test space $\mathcal{V}$ of Assumption~\ref{ass:weak}. The
	standard $L^2$ operator-theoretic framework for the Koopman
	operator and its generator
	\cite{eisner2015operator,lasota1994chaos} is not required for
	these identities.
	
	For control-affine systems
	\begin{equation}
		\dot{\mathbf{x}}
		= \mathbf{f}_0(\mathbf{x})
		+ \sum_{j=1}^m \mathbf{g}_j(\mathbf{x})\,u_j,
		\label{eq:control_affine}
	\end{equation}
	with autonomous drift $\mathbf{f}_0:\mathbb{R}^n\to\mathbb{R}^n$
	and control vector fields $\mathbf{g}_j$, fixing a constant
	input $\mathbf{u}$ makes the forced system autonomous, and
	its Koopman generator is the directional derivative along that
	vector field:
	\begin{equation}
		\mathcal{K}_{\mathbf{u}} f(\mathbf{x})
		=
		\nabla f(\mathbf{x})^\intercal
		\left(
		\mathbf{f}_0(\mathbf{x})
		+\sum_{j=1}^m \mathbf{g}_j(\mathbf{x})u_j
		\right).
		\label{eq:fixed_input_generator}
	\end{equation}
	For the pH system~\eqref{eq:ph_nonlinear}, this specializes to
	$\mathbf{f}_0=\mathbf{v}_0$ with the columns of $\mathbf{G}$ as
	control vector fields.
	
	If $\mathbf{u}=\mathbf{u}(t)$ varies with time,
	$\mathcal{K}_{\mathbf{u}(t)}$ is a time-dependent family of
	directional derivatives rather than the generator of a single
	autonomous Koopman semigroup. An autonomous semigroup is recovered by fixing the input, fixing
	a feedback law $\mathbf{u} = \mathbf{u}(\mathbf{x})$, or
	augmenting the state by including the time variable, with
	$\dot{t} = 1$, so that $\tilde{\mathbf{x}}=(\mathbf{x},t)^\intercal$
	evolves autonomously on $\mathcal{X}\times\mathbb{R}$.

\subsection{Reference, Sampling, and Empirical Measures}

Three measures play distinct roles in the sequel. For a measure
$d\mu=\rho\,d\mathbf{x}$, define the \emph{weighted divergence} as follows:
\begin{equation}
	\operatorname{div}_{\mu}\mathbf{v}
	\triangleq \rho^{-1}\operatorname{div}(\rho\mathbf{v}).
	\label{eq:weighted_divergence}
\end{equation}

The first measure is a conservative reference measure $\mu_J$.
It is used in Section~\ref{sec:decomposition} to express the
conservative part of the generator as a skew operator, the
infinite-dimensional analogue of a skew-symmetric matrix. The
requirement is that the conservative flow
$\mathbf{v}_J = \mathbf{J}\nabla\mathcal{H}$ preserve the
weighting $\mu_J$: as the flow moves points through state space,
it neither concentrates nor disperses the measure. Writing
$\mu_J = \rho\,d\mathbf{x}$, this is the condition
\begin{equation}
	\operatorname{div}_{\mu_J}
	\bigl(\mathbf{J}\nabla\mathcal{H}\bigr)=0,
	\label{eq:muJ_divergence}
\end{equation}
with the boundary conditions of Assumption~\ref{ass:weak}. Many
weightings satisfy this condition, and the skewness identity
holds for any of them. When $\mathbf{J}$ is constant, uniform
(Lebesgue) weighting is the simplest valid choice. Note that
$\mu_J$ describes only the conservative flow, not to
the full pH system with dissipation.

The second measure is a sampling measure $\nu$, used in
Section~\ref{sec:approximation} to define the inner product
$\langle f, g\rangle_\nu = \int_{\mathcal{X}} f g \, d\nu$ and
some weak projection matrices built from it. Whereas $\mu_J$ must
be preserved by the conservative flow, $\nu$ is unconstrained
and may be chosen freely, e.g., a quadrature measure, a user-chosen
sampling distribution, or a measure supported on a compact
region of interest. The role of $\nu$ and
the conditions it must meet for the conservative matrix to be
skew, are developed in Section~\ref{sec:approximation}; the
examples in Section~\ref{sec:example} use concrete choices.

The third measure is the empirical measure of sampled states
$\mathbf{x}_1, \ldots, \mathbf{x}_{N_{\rm samp}}$,
\begin{equation}
	\widehat{\nu}_{N_{\rm samp}}
	\triangleq
	\frac{1}{N_{\rm samp}}
	\sum_{k=1}^{N_{\rm samp}} \delta_{\mathbf{x}_k},
	\label{eq:empirical_measure}
\end{equation}
where $\delta_{\mathbf{x}_k}$ is the Dirac measure at
$\mathbf{x}_k$. Replacing $\nu$ by $\widehat{\nu}_{N_{\rm samp}}$
turns the inner product $\langle f, g\rangle_\nu$ into the sample
average $\frac{1}{N_{\rm samp}} \sum_k f(\mathbf{x}_k)
g(\mathbf{x}_k)$, which estimates the population value with
finite-sample error and should not be treated as exact. We
treat these empirical matrices as approximations and enforce pH
structure through constraints rather than statistical
convergence; a quantitative analysis of the sampling error is
left to future work.

The three measures are kept distinct because a dissipative pH
flow does not supply a single natural measure that serves all
three roles. Its invariant measures are typically unsuitable for
projection: they may be singular or degenerate, such as a Dirac
mass at a stable equilibrium, and cannot weight the region of
state space where a useful approximation is needed. The
conservative reference measure, the sampling measure, and the
empirical measure are therefore introduced separately.

\section{Koopman Generator Decomposition}
\label{sec:decomposition}

We now establish the Koopman-level identities induced by the pH
vector-field decomposition. The results come in two parts of
different character. First, Proposition~\ref{prop:decomposition} gives the bracket
decomposition~\eqref{eq:gen_split}, a pointwise identity on
smooth observables. Specializing it to the Hamiltonian observable
yields the energy identity~\eqref{eq:koopman_ph_energy}, which
recovers the pH passivity inequality. Second, Proposition~\ref{thm:skewness} refines the conservative
part. When the reference measure $\mu_J$ is preserved by the
conservative flow, the
derivation $\mathcal{K}_J$ is skew-symmetric in $L^2(\mu_J)$,
as stated in~\eqref{eq:skew_result}. Unlike
Proposition~\ref{prop:decomposition}, which is pointwise, this identity holds in
$L^2(\mu_J)$ and uses the integration-by-parts formula of
Assumption~\ref{ass:weak}. The choice of $\mu_J$ is discussed
further in Remark~\ref{rem:conservative_measure_condition}.
	
\begin{proposition}[pH bracket decomposition of the Koopman generator]
	\label{prop:decomposition}
	Consider the pH system~\eqref{eq:ph_nonlinear} under
	Assumption~\ref{ass:pointwise}. For a smooth observable $f$,
	the controlled Koopman generator decomposes as
	\begin{equation}
		\mathcal{K}_{\mathbf{u}}f
		= \{f,\mathcal{H}\}_J - [f,\mathcal{H}]_R
		+ \nabla f^\intercal\mathbf{G}\mathbf{u},
		\label{eq:gen_split}
	\end{equation}
	equivalently, $\mathcal{K}_{\mathbf{u}} = \mathcal{K}_J -
	\mathcal{K}_R + \mathcal{K}_{\rm port}(\mathbf{u})$. Thus, the
	Hamiltonian observable satisfies the following energy identity and passivity inequality:
	\begin{equation}
		\mathcal{K}_{\mathbf{u}}\mathcal{H}
		= -[\mathcal{H},\mathcal{H}]_R + \mathbf{y}^\intercal\mathbf{u}
		\le \mathbf{y}^\intercal\mathbf{u},
		\label{eq:koopman_ph_energy}
	\end{equation}
	since $[\mathcal{H},\mathcal{H}]_R \ge 0$ and
	$\{\mathcal{H},\mathcal{H}\}_J = 0$.
\end{proposition}

\begin{proof}[Proof of Proposition~\ref{prop:decomposition}]
	Applying the Koopman generator along the controlled pH vector
	field and expanding,
	\begin{align*}
		\mathcal{K}_{\mathbf{u}}f
		&=
		\nabla f^\intercal
		\left[
		(\mathbf{J}-\mathbf{R})\,\nabla\mathcal{H}
		+\mathbf{G}\mathbf{u}
		\right] \notag\\
		&=
		\{f,\mathcal{H}\}_J
		-
		[f,\mathcal{H}]_R
		+
		\nabla f^\intercal\mathbf{G}\mathbf{u},
	\end{align*}
	which gives~\eqref{eq:gen_split}. Skew-symmetry of $\mathbf{J}$
	gives $\{\mathcal{H},\mathcal{H}\}_J =
	(\nabla\mathcal{H})^\intercal\mathbf{J}\,\nabla\mathcal{H} = 0$,
	and positive semidefiniteness of $\mathbf{R}$ gives
	$[f,f]_R = (\nabla f)^\intercal\mathbf{R}\,\nabla f \ge 0$
	pointwise for every smooth $f$. Substituting
	$f=\mathcal{H}$ into~\eqref{eq:gen_split}
	and using $\{\mathcal{H},\mathcal{H}\}_J = 0$,
	\begin{align}
		\mathcal{K}_{\mathbf{u}}\mathcal{H}
		&=
		-[\mathcal{H},\mathcal{H}]_R
		+
		(\nabla\mathcal{H})^\intercal\mathbf{G}\mathbf{u}
		\notag\\
		&=
		-[\mathcal{H},\mathcal{H}]_R
		+
		\mathbf{y}^\intercal\mathbf{u}
		\le
		\mathbf{y}^\intercal\mathbf{u},
		\notag
	\end{align}
	which is~\eqref{eq:koopman_ph_energy}; the second equality uses
	$\mathbf{y}=\mathbf{G}^\intercal\nabla\mathcal{H}$ and the
	inequality uses $[\mathcal{H},\mathcal{H}]_R\ge 0$.
\end{proof}

\begin{proposition}[Skewness of the conservative derivation]
	\label{thm:skewness}
	Let $\mu_J$ be a reference measure satisfying
	\begin{equation}
		\operatorname{div}_{\mu_J}(\mathbf{J}\,\nabla\mathcal{H}) = 0,
		\label{eq:thm_muJ_condition}
	\end{equation}
	and let $\mathcal{V}$ be a test function space satisfying
	Assumption~\ref{ass:weak}. Then, the conservative derivation
	$\mathcal{K}_J$ from Proposition~\ref{prop:decomposition} is
	skew-symmetric on $\mathcal{V}$ in the $L^2(\mu_J)$ inner
	product, i.e.,
	\begin{equation}
		\langle g,\mathcal{K}_J f\rangle_{L^2(\mu_J)}
		= -\langle \mathcal{K}_J g, f\rangle_{L^2(\mu_J)},
		\quad f,g \in \mathcal{V}.
		\label{eq:skew_result}
	\end{equation}
\end{proposition}

\begin{proof}[Proof of Proposition~\ref{thm:skewness}]
	Let $\mathbf{v}_J=\mathbf{J}\,\nabla\mathcal{H}$. For
	$f,g\in\mathcal{V}$, the product rule gives
	\begin{equation}
		\nabla(fg)^\intercal \mathbf{v}_J
		= g\, \nabla f^\intercal \mathbf{v}_J
		+ f\, \nabla g^\intercal \mathbf{v}_J
		= g\,\mathcal{K}_J f + f\,\mathcal{K}_J g.
		\label{eq:product_rule}
	\end{equation}
	The integration-by-parts formula of
	Assumption~\ref{ass:weak}, applied to $\mathbf{v}_J$, together
	with the conservative-measure
	condition~\eqref{eq:thm_muJ_condition}, gives
	\begin{equation}
		\int_{\mathcal{X}}
		\nabla(fg)^\intercal\mathbf{v}_J\,d\mu_J
		=
		-\int_{\mathcal{X}}
		fg\,\operatorname{div}_{\mu_J}(\mathbf{v}_J)\,d\mu_J
		= 0.
		\label{eq:ibp_zero}
	\end{equation}
	Combining~\eqref{eq:product_rule} and~\eqref{eq:ibp_zero},
	\begin{equation}
		\int_{\mathcal{X}}
		\left(g\,\mathcal{K}_J f+f\,\mathcal{K}_J g\right)d\mu_J
		= 0,
		\label{eq:div_zero}
	\end{equation}
	which is~\eqref{eq:skew_result}.
\end{proof}

\begin{remark}[The decomposition and its conservative refinement]
	\label{rem:story}
	Proposition~\ref{prop:decomposition} and
	Proposition~\ref{thm:skewness} together describe how the pH
	structure of the vector field appears in the Koopman generator.
	Proposition~\ref{prop:decomposition} is algebraic and
	unconditional: the generator always splits into conservative,
	dissipative, and port derivations, and the Hamiltonian
	observable always obeys the energy identity.
	Proposition~\ref{thm:skewness} refines the conservative part: under
	a measure compatible with the conservative flow, that part is
	not merely a derivation but a skew-symmetric operator, the
	infinite-dimensional analogue of the skew matrix $\mathbf{J}$.
	The dissipative part admits no such operator-level positivity;
	Remark~\ref{rem:metric_vs_derivation} develops this asymmetry.
	These two results motivate the finite-dimensional constructions
	of Section~\ref{sec:approximation}, where the conservative
	matrix is skew and the dissipative matrix is positive
	semidefinite.
	Two further theorems sharpen this decomposition: a weak-form
	skewness result with an exactly verifiable Gibbs instance,
	developed in Section~\ref{sec:approximation}, and a
	residual-aware closed-loop bound, developed in
	Section~\ref{sec:control}.
\end{remark}

\begin{remark}[Metric-bracket positivity versus positivity of a derivation]
	\label{rem:metric_vs_derivation}
	Proposition~\ref{prop:decomposition} asserts positivity of the
	metric bracket on matching arguments, $[f,f]_R \ge 0$, not
	positivity of the first-order derivation $\mathcal{K}_R f =
	[f,\mathcal{H}]_R$. The derivation can take either sign on
	arbitrary observables: for a reference measure $\eta$,
	\begin{equation}
		\langle f,\,\mathcal{K}_R f\rangle_\eta
		=
		\int_{\mathcal{X}} f [f,\mathcal{H}]_R\,d\eta,
		\label{eq:KR_no_fixed_sign}
	\end{equation}
	which has no fixed sign for arbitrary $f$. Positivity enters
	the dissipative story only through the bracket on matching
	arguments and through its specialization
	$-[\mathcal{H},\mathcal{H}]_R \le 0$ in the energy identity
	$\mathcal{K}_{\mathbf{u}}\mathcal{H} = -[\mathcal{H},\mathcal{H}]_R
	+ \mathbf{y}^\intercal\mathbf{u}$. 
\end{remark}
	
	\smallskip
	\noindent\emph{A one-dimensional counterexample.}
	The following example clarifies the distinction. Consider the
	scalar stable pH system
	\begin{equation}
		\begin{gathered}
			\dot{x}=-x,\qquad
			\mathcal{H}(x)=\frac{1}{2}x^2,\\
			J=0,\qquad R=1,\qquad G=0.
		\end{gathered}
		\label{eq:one_dim_counter_system}
	\end{equation}
	Then $\nabla\mathcal{H}(x)=x$,
	$\mathcal{K}_R f=xf'(x)$, and
	$\mathcal{K}_0 f=-xf'(x)$. Let
	\begin{equation}
		f(x)=
		\begin{cases}
			\exp\!\left(-\dfrac{x^2}{1-x^2}\right), & |x|<1,\\[1ex]
			0, & |x|\ge 1.
		\end{cases}
		\label{eq:bell_observable_counter}
	\end{equation}
	This is a standard smooth bump function with maximum at the
	origin. For $|x|<1$,
	\begin{equation}
		\begin{aligned}
			f'(x)
			&=
			-\frac{2x}{(1-x^2)^2}f(x),
			\\
			\mathcal{K}_R f(x)
			&=
			-\frac{2x^2}{(1-x^2)^2}f(x)
			\le 0.
		\end{aligned}
		\label{eq:bell_KR_negative}
	\end{equation}
	Consequently,
	\begin{equation}
		f(x)\mathcal{K}_R f(x)
		=
		-\frac{2x^2}{(1-x^2)^2}f(x)^2
		\le 0,
		\label{eq:bell_pointwise_product}
	\end{equation}
	and, with respect to Lebesgue measure,
	\begin{equation}
		\langle f,\mathcal{K}_R f\rangle
		=
		-\int_{-1}^{1}
		\frac{2x^2}{(1-x^2)^2}f(x)^2\,dx
		<0.
		\label{eq:bell_inner_product_negative}
	\end{equation}
	The integral is finite because $f$ is smooth and compactly
	supported on $[-1,1]$.
	The same strict negativity holds for any measure assigning
	positive mass to a subset of $(-1,1)\setminus\{0\}$ on which the
	integral is evaluated. Thus $\mathcal{K}_R$ is not positive on
	arbitrary observables. For the same observable, however, the
	metric bracket is nonnegative:
	\begin{equation}
		[f,f]_R=(f'(x))^2
		=
		\frac{4x^2}{(1-x^2)^4}f(x)^2
		\ge 0
		\qquad (|x|<1).
		\label{eq:bell_metric_positive}
	\end{equation}
	For the Hamiltonian observable,
	\begin{equation}
		[\mathcal{H},\mathcal{H}]_R=x^2\ge 0,
		\qquad
		\mathcal{K}_0\mathcal{H}=-x^2\le 0.
		\label{eq:counter_hamiltonian_dissipation}
	\end{equation}
	Thus the pH system dissipates the Hamiltonian, while an arbitrary
	bell-shaped observable centered at the stable equilibrium can
	increase along trajectories. Indeed, for
	$x(t)=e^{-t}x_0$ with $0<|x_0|<1$, the state moves toward the
	maximum of the bell-shaped observable, so $f(x(t))$ can increase
	even though $\mathcal{H}(x(t))$ decreases. This counterexample shows why dissipative positivity must be located in the metric bracket and in the energy balance on the Hamiltonian, and not in the first-order derivation on arbitrary observables.

	\begin{remark}[Uniqueness and converse questions]
		\label{rem:uniqueness_converse}
		\emph{Nonuniqueness.}
		The representation
		\begin{equation}
			\mathbf{v}
			=
			(\mathbf{J}-\mathbf{R})\nabla\mathcal{H}
			\label{eq:ph_vector_field_nonunique}
		\end{equation}
		is generally not unique. The storage function itself may be a
		modeling choice, and even for fixed $\mathcal{H}$ the fields
		$\mathbf{J}$ and $\mathbf{R}$ are constrained only through their
		action on $\nabla\mathcal{H}$. Components that annihilate
		$\nabla\mathcal{H}$ do not change the vector field. Additional
		physical interconnection structure, sparsity, coordinate choices,
		or parameterizations are therefore needed for uniqueness.
	\end{remark}

	\begin{remark}[Converse constructions]
		\label{rem:converse}
		\emph{Pointwise local converse.}
		Away from critical points of $\mathcal{H}$, let
		$g=\nabla\mathcal{H}(\mathbf{x})\neq 0$ and suppose a smooth drift
		$\mathbf{v}$ satisfies $g^\intercal\mathbf{v}\le 0$ at the point.
		Writing
		\[
		\mathbf{v}_\parallel
		=
		\frac{g^\intercal\mathbf{v}}{\|g\|^2}g,
		\qquad
		\mathbf{v}_\perp
		=
		\mathbf{v}-\mathbf{v}_\parallel,
		\]
		one possible local construction is
		\begin{equation}
			\mathbf{R}_0
			=
			-\frac{g^\intercal\mathbf{v}}{\|g\|^4}gg^\intercal
			\succeq0,
			\qquad
			\mathbf{J}_0
			=
			\frac{\mathbf{v}_\perp g^\intercal
				-g\mathbf{v}_\perp^\intercal}{\|g\|^2},
			\label{eq:local_ph_converse}
		\end{equation}
		which gives $\mathbf{J}_0g-\mathbf{R}_0g=\mathbf{v}$.
		This is only a pointwise observation. Smoothness across critical
		points of $\mathcal{H}$ is a separate requirement, and the displayed formula is undefined at $\nabla\mathcal{H}=0$, so the pointwise construction can fail to be continuous at critical points of $\mathcal{H}$.
		Compatibility with prescribed $\mathbf{J}$, $\mathbf{R}$, and
		$\mathbf{G}$, and global boundary or topology constraints are also
		separate requirements.
		
		\emph{Operator-triple converse.}
		Conversely, an abstract triple of operators with skew-looking and
		positive-looking algebra need not come from a pH vector field. The
		components must be first-order derivations induced by vector
		fields. The interconnection derivation must arise from a skew
		matrix field $\mathbf{J}$, the metric bracket must arise from a
		symmetric positive semidefinite matrix field $\mathbf{R}$, and the
		port term must arise from $\mathbf{G}\mathbf{u}$. Abstract adjointness or
		matrix skewness alone is insufficient to reconstruct a nonlinear
		pH realization.
		
		\emph{Linear case.}
		In the linear finite-dimensional case this ambiguity is reduced
		only after fixing the storage metric. If
		\[
		\dot{x}=Ax+Bu,
		\qquad
		\mathcal{H}(x)=\frac12x^\intercal Qx,
		\qquad
		Q=Q^\intercal\succ0,
		\]
		and $A=(J-R)Q$ is sought with $J=-J^\intercal$ and
		$R=R^\intercal\succeq0$, then for this fixed $Q$ one must have
		\begin{align}
			J
			&=
			\frac12\left(AQ^{-1}-Q^{-1}A^\intercal\right),
			\notag\\
			R
			&=
			-\frac12\left(AQ^{-1}+Q^{-1}A^\intercal\right),
			\label{eq:linear_fixed_Q_decomposition}
		\end{align}
		provided the resulting $R$ is positive semidefinite. For this
		fixed $Q$, the condition $R\succeq0$ is equivalently
		\begin{equation}
			AQ^{-1}+Q^{-1}A^\intercal\preceq0,
			\qquad
			QA+A^\intercal Q\preceq0.
			\label{eq:linear_Q_dissipativity}
		\end{equation}
		Thus the fixed storage metric $Q$ determines the corresponding
		linear pH dissipativity condition. For
		a lifted linear surrogate $A_z=(S-D)P$ with fixed
		$P=P^\intercal\succ0$,
		\begin{align}
			S
			&=
			\frac12\left(A_zP^{-1}-P^{-1}A_z^\intercal\right),
			\notag\\
			D
			&=
			-\frac12\left(A_zP^{-1}+P^{-1}A_z^\intercal\right),
			\label{eq:lifted_fixed_P_decomposition}
		\end{align}
		provided $D\succeq0$. Without fixing the storage metric, such
		representations are generally nonunique.
		
		\emph{Consequences for approximation.}
		Finite-dimensional skew or positive-looking matrices do not by
		themselves identify a unique nonlinear pH realization. The
		dictionary, storage metric, residual, and port representation are
		part of the model specification.
	\end{remark}

	\section{Finite-Dimensional Structure-Preserving Approximations}
	\label{sec:approximation}
	
	Section~\ref{sec:decomposition} established pointwise generator
	identities and weak conservative skewness under a
	conservative-measure condition. This section turns those
	identities into finite-dimensional objects: a dictionary and mass
	matrix, weak matrices for $\mathcal{K}_J$, $\mathcal{K}_R$, and
	the metric bracket, a projection residual $r_N$, and finally a
	lifted pH surrogate $(S,D,P,B)$. Proposition~\ref{prop:splitting}
	is the key structural statement: under the appropriate measure and
	boundary assumptions, the conservative weak matrix is skew while
	the dissipative metric-bracket matrix is positive semidefinite.
	
	There are several levels of claim. Pointwise identities give
	$\mathcal{K}_{\mathbf{u}}f=\{f,\mathcal{H}\}_J-[f,\mathcal{H}]_R
	+\nabla f^\intercal\mathbf{G}\mathbf{u}$ and, for the Hamiltonian,
	$\mathcal{K}_{\mathbf{u}}\mathcal{H}
	=-[\mathcal{H},\mathcal{H}]_R+\mathbf{y}^\intercal\mathbf{u}$.
	Galerkin weak skewness requires the stated measure and boundary
	assumptions. In contrast, $D_R\succeq0$ requires only a nonnegative
	measure, $\mathbf R\succeq0$, and finite entries. Empirical estimates approximate those population
	objects and may require imposed constraints. The lifted pH
	surrogate is passive in lifted coordinates, while plant-level
	conclusions require residual, injectivity, and output-consistency
	assumptions.
	
	\subsection{Dictionary, Mass Matrix, and Weak Projection}
	
	Let
	\begin{equation}
		\Psi(\mathbf{x})
		=
		\begin{bmatrix}
			\psi_1(\mathbf{x}) & \cdots & \psi_N(\mathbf{x})
		\end{bmatrix}^{\intercal},
		\qquad
		z=\Psi(\mathbf{x}),
		\label{eq:dictionary_lift}
	\end{equation}
	where $\psi_i\in\mathcal{V}$, or more generally the
	dictionary functions are smooth and satisfy the boundary
	assumptions of Section~\ref{sec:prelim}. Define
	$\mathcal{V}_N=\operatorname{span}\{\psi_1,\ldots,\psi_N\}$.
	Let $\nu$ be the Galerkin or sampling measure introduced in
	Section~\ref{sec:prelim}, and write the population inner
	product as
	\begin{equation}
		\langle f,g\rangle_\nu
		\triangleq
		\int_{\mathcal{X}} f(\mathbf{x})g(\mathbf{x})\,
		d\nu(\mathbf{x}).
		\label{eq:nu_inner_product}
	\end{equation}
	The mass matrix is
	\begin{equation}
		M_{ij}=\langle\psi_i,\psi_j\rangle_\nu.
		\label{eq:mass_matrix}
	\end{equation}
	We assume $M\succ0$, equivalently the dictionary functions
	are linearly independent in $L^2(\nu)$. For a derivation
	$\mathcal{L}$, the Galerkin weak-form matrix is
	\begin{equation}
		(B_{\mathcal{L}})_{ij}
		=
		\langle\psi_i,\mathcal{L}\psi_j\rangle_\nu,
		\qquad
		A_{\mathcal{L}}=M^{-1}B_{\mathcal{L}}.
		\label{eq:weak_galerkin_matrix}
	\end{equation}
	The matrix $A_{\mathcal{L}}$ is the coefficient-space
	operator. Only when the dictionary is orthonormal in
	$L^2(\nu)$ does $M=I$ and $B_{\mathcal{L}}=A_{\mathcal{L}}$.
	
	\begin{remark}[Coefficient and coordinate conventions]
		\label{rem:coefficient_coordinate}
		With the convention $z=\Psi(\mathbf{x})$ as a column vector,
		$A_{\mathcal{L}}$ acts on coefficient vectors of observables
		$f_c=\Psi^\intercal c$. If
		\[
		\mathcal{L}\psi_j\approx
		\sum_i (A_{\mathcal{L}})_{ij}\psi_i,
		\]
		then the lifted coordinate vector satisfies
		\[
		\dot z=\mathcal{L}\Psi(\mathbf{x})
		\approx A_{\mathcal{L}}^\intercal z.
		\]
		Thus $A_{\mathcal{L}}$ is a coefficient-space operator, whereas
		$A_{\mathcal{L}}^\intercal$ is the associated coordinate-dynamics
		matrix. For the conservative part, if
		$A_J^\intercal M+MA_J=0$, then $A_z=A_J^\intercal$ satisfies
		\begin{equation}
			A_z^\intercal M^{-1}+M^{-1}A_z=0.
			\label{eq:coordinate_mass_skew}
		\end{equation}
		Equivalently, the conservative coordinate-space realization can
		be written with $P=M^{-1}$ and $S=-B_J$, since
		$A_z=A_J^\intercal=-B_JM^{-1}=SP$.
	\end{remark}
	
	\subsection{Conservative Weak Form and Dissipative Dirichlet Matrix}
	
		Let the sampling measure have a density,
	$\nu=\rho\,d\mathbf{x}$ with $\rho\in C^1(\mathcal{X})$ and
	$\rho>0$. The conservative weak matrix below is exactly skew
	when the weighted conservative flow $\rho\mathbf{v}_J$, with
	$\mathbf{v}_J=\mathbf{J}\nabla\mathcal{H}$, is divergence free
	and its boundary flux through the dictionary products vanishes.
	Lemma~\ref{lem:skew_defect} makes this precise and, for a
	generic measure, converts the residual skewness into a
	computable defect.
	
	For the conservative derivation, define
	\begin{align}
		(B_J)_{ij}
		&=
		\langle\psi_i,\mathcal{K}_J\psi_j\rangle_\nu
		\nonumber\\
		&=
		\int_{\mathcal{X}}
		\psi_i\{\psi_j,\mathcal{H}\}_J\,d\nu,
		\qquad
		A_J=M^{-1}B_J .
		\label{eq:BJ_AJ}
	\end{align}
	The matrix of the first-order dissipative derivation is
	\begin{equation}
		(C_R)_{ij}
		=
		\langle\psi_i,\mathcal{K}_R\psi_j\rangle_\nu
		=
		\int_{\mathcal{X}}
		\psi_i[\psi_j,\mathcal{H}]_R\,d\nu .
		\label{eq:CR_derivation_matrix}
	\end{equation}
	The matrix $C_R$ is generally neither symmetric nor positive
	semidefinite. The positive semidefinite finite-dimensional
	dissipative object is instead the Dirichlet, or
	metric-bracket, matrix
	\begin{equation}
		(D_R)_{ij}
		=
		\int_{\mathcal{X}}[\psi_i,\psi_j]_R\,d\nu
		=
		\int_{\mathcal{X}}
		\nabla\psi_i^\intercal\mathbf{R}\nabla\psi_j\,d\nu .
		\label{eq:DR_dirichlet_matrix}
	\end{equation}
	Thus $D_R$ represents the symmetric metric bracket, not the Galerkin matrix of the first-order derivation $\mathcal{K}_R$. In the finite element analogy, $M$ is the mass matrix and $D_R$ is a weighted stiffness matrix.
	
		\begin{lemma}[Skewness defect identity]
		\label{lem:skew_defect}
		Let $\nu=\rho\,d\mathbf{x}$ with $\rho\in C^1(\mathcal{X})$,
		$\rho>0$, and let $\psi_i,\psi_j\in C^1(\mathcal{X})$ have
		integrable products with $\rho\mathbf{v}_J$. With
		$\mathbf{v}_J=\mathbf{J}\nabla\mathcal{H}$ and $B_J$ as
		in~\eqref{eq:BJ_AJ},
		\begin{equation}
			(B_J+B_J^\intercal)_{ij}
			=
			\oint_{\partial\mathcal{X}}
			\psi_i\psi_j\,\rho\,\mathbf{v}_J\cdot\mathbf{n}\,dS
			-
			\int_{\mathcal{X}}
			\psi_i\psi_j\operatorname{div}(\rho\mathbf{v}_J)\,d\mathbf{x}.
			\label{eq:skew_defect}
		\end{equation}
		Hence $B_J$ is skew-symmetric if and only if the right-hand
		side vanishes for every pair $i,j$. Two sufficient conditions
		are $\operatorname{div}(\rho\mathbf{v}_J)=0$ on $\mathcal{X}$
		together with vanishing boundary flux for all dictionary
		products.
	\end{lemma}
	\begin{proof}
		By~\eqref{eq:BJ_AJ},
		$(B_J+B_J^\intercal)_{ij}
		=\int_{\mathcal{X}}(\psi_i\nabla\psi_j+\psi_j\nabla\psi_i)^\intercal
		\mathbf{v}_J\,\rho\,d\mathbf{x}
		=\int_{\mathcal{X}}\nabla(\psi_i\psi_j)^\intercal
		(\rho\mathbf{v}_J)\,d\mathbf{x}$. The product rule gives
		$\operatorname{div}(\psi_i\psi_j\,\rho\mathbf{v}_J)
		=\nabla(\psi_i\psi_j)^\intercal(\rho\mathbf{v}_J)
		+\psi_i\psi_j\operatorname{div}(\rho\mathbf{v}_J)$, and the
		divergence theorem applied to the field
		$\psi_i\psi_j\,\rho\mathbf{v}_J$ yields~\eqref{eq:skew_defect}.
	\end{proof}

	\begin{proposition}[Galerkin weak-form structure of the pH brackets]
		\label{prop:splitting}
		Let $\Psi=(\psi_1,\ldots,\psi_N)^\intercal$ with
		$\psi_i\in\mathcal{V}\cap C^1(\mathcal{X})$, and let $M$,
		$B_J$, $A_J$, $C_R$, and $D_R$ be defined as above for a
		measure $\nu=\rho\,d\mathbf{x}$ with $\rho\in C^1(\mathcal{X})$,
		$\rho>0$. Assume $M\succ0$, that
		$\operatorname{div}(\rho\mathbf{v}_J)=0$ on $\mathcal{X}$ with
		$\mathbf{v}_J=\mathbf{J}\nabla\mathcal{H}$, and that the
		boundary flux
		$\oint_{\partial\mathcal{X}}\psi_i\psi_j\,\rho\,
		\mathbf{v}_J\cdot\mathbf{n}\,dS$ vanishes for all products of
		dictionary functions. Then
		\begin{equation}
			B_J=-B_J^\intercal .
			\label{eq:BJ_skew}
		\end{equation}
		Consequently,
		\begin{equation}
			A_J^\intercal M+M A_J=0 .
			\label{eq:AJ_mass_skew}
		\end{equation}
		Moreover,
		\begin{equation}
			D_R=D_R^\intercal\succeq0 .
			\label{eq:DR_psd}
		\end{equation}
		These statements do not imply that $C_R$ is symmetric or
		positive semidefinite.
	\end{proposition}
	\begin{proof}
		Under the stated hypotheses both terms on the right-hand side
		of the defect identity~\eqref{eq:skew_defect} vanish for every
		pair, so $B_J=-B_J^\intercal$ by Lemma~\ref{lem:skew_defect}.
		Since $A_J=M^{-1}B_J$, we have $MA_J=B_J$ and hence
		\begin{equation}
			A_J^\intercal M+M A_J
			=
			B_J^\intercal+B_J
			=
			0 .
			\label{eq:mass_skew_proof}
		\end{equation}
		For any $c\in\mathbb{R}^N$, let
		$\phi_c=\sum_{i=1}^N c_i\psi_i$. Then
		\begin{equation}
			c^\intercal D_R c
			=
			\int_{\mathcal{X}}
			\nabla\phi_c^\intercal\mathbf{R}\nabla\phi_c\,d\nu
			=
			\int_{\mathcal{X}}[\phi_c,\phi_c]_R\,d\nu
			\ge0 .
			\label{eq:DR_psd_proof}
		\end{equation}
		Symmetry of $D_R$ follows from symmetry of $\mathbf{R}$. This
		argument concerns the metric-bracket matrix $D_R$; it does not
		prove any sign property of the first-order derivation matrix
		$C_R$.
	\end{proof}

	\begin{remark}[Minimal skewness condition]
		\label{rem:minimal_skewness}
		The pairwise condition
		\begin{equation}
			\begin{gathered}
				\int_{\mathcal{X}}
				\nabla(\psi_i\psi_j)^\intercal\mathbf{v}_J\,d\nu=0,\\
				\text{for all } i,j=1,\ldots,N,
			\end{gathered}
			\label{eq:nu_weak_conservative}
		\end{equation}
		is the minimal algebraic requirement for exact skewness of
		$B_J$, since it is the vanishing of the right-hand side
		of~\eqref{eq:skew_defect}. Lemma~\ref{lem:skew_defect} shows
		that $\operatorname{div}(\rho\mathbf{v}_J)=0$ with vanishing
		boundary flux implies it, and for a generic measure it bounds
		the skewness defect,
		\begin{equation}
			\begin{aligned}
				|(B_J+B_J^\intercal)_{ij}|
				\le{}&
				\int_{\mathcal{X}}|\psi_i\psi_j|\,
				|\operatorname{div}(\rho\mathbf{v}_J)|\,d\mathbf{x}\\
				&+
				\Bigl|\oint_{\partial\mathcal{X}}
				\psi_i\psi_j\,\rho\,\mathbf{v}_J\cdot\mathbf{n}\,dS\Bigr|,
			\end{aligned}
			\label{eq:skew_defect_bound}
		\end{equation}
		through the divergence defect and the boundary flux.
	\end{remark}

	\begin{lemma}[Energy-weighted and unimodular conservative measures]
		\label{lem:unimodular}
		Let $\mathbf{v}_J=\mathbf{J}\nabla\mathcal{H}$.
		\emph{(i)} If $\mathbf{J}$ is constant and
		$\rho=g(\mathcal{H})$ with $g\in C^1$, $g>0$, then
		\begin{equation}
			\operatorname{div}(\rho\mathbf{v}_J)
			=
			g'(\mathcal{H})\{\mathcal{H},\mathcal{H}\}_J
			+
			g(\mathcal{H})\operatorname{tr}(\mathbf{J}\nabla^2\mathcal{H})
			=0 ,
			\label{eq:energy_weighted_div}
		\end{equation}
		since $\{\mathcal{H},\mathcal{H}\}_J=0$ and
		$\operatorname{tr}(\mathbf{J}\nabla^2\mathcal{H})=0$ for skew
		$\mathbf{J}$. \emph{(ii)} For state-dependent
		$\mathbf{J}(\mathbf{x})$,
		\begin{equation}
			\operatorname{div}(\rho\mathbf{v}_J)
			=
			\sum_j\Bigl(\sum_i\partial_i(\rho J_{ij})\Bigr)
			\partial_j\mathcal{H}
			+
			\rho\operatorname{tr}(\mathbf{J}\nabla^2\mathcal{H}),
			\label{eq:state_dependent_div}
		\end{equation}
		and the second term vanishes by skew-symmetry. Hence the
		Hamiltonian-independent condition
		\begin{equation}
			\sum_i\partial_i(\rho J_{ij})=0
			\quad\text{for all } j
			\label{eq:unimodular_condition}
		\end{equation}
		implies $\operatorname{div}(\rho\mathbf{v}_J)=0$ for every
		$\mathcal{H}\in C^2$.
	\end{lemma}
	\begin{proof}
		Expand
		$\operatorname{div}(\rho\mathbf{v}_J)
		=\sum_k\partial_k\bigl(\rho\sum_l J_{kl}\partial_l\mathcal{H}\bigr)$
		by the product rule. The terms
		$\rho\sum_{k,l}J_{kl}\partial_k\partial_l\mathcal{H}
		=\rho\operatorname{tr}(\mathbf{J}\nabla^2\mathcal{H})$ vanish
		because $\mathbf{J}$ is skew and $\nabla^2\mathcal{H}$ is
		symmetric. For constant $\mathbf{J}$ the remaining term is
		$g'(\mathcal{H})\nabla\mathcal{H}^\intercal\mathbf{J}
		\nabla\mathcal{H}=g'(\mathcal{H})\{\mathcal{H},\mathcal{H}\}_J=0$;
		for state-dependent $\mathbf{J}$ it is
		$\sum_j(\sum_i\partial_i(\rho J_{ij}))\partial_j\mathcal{H}$.
	\end{proof}

	\begin{remark}[Conservative-measure condition]
		\label{rem:conservative_measure_condition}
		The condition~\eqref{eq:muJ_divergence} on $\mu_J$ is a
		conservative-flow condition, not an invariant-measure
		assumption for the full dissipative pH dynamics, and the
		sampling measures used later need not be invariant. For
		canonical constant $\mathbf{J}$, Lebesgue weighting is the
		simplest valid choice, since
		\begin{equation}
			\operatorname{div}(\mathbf{J}\nabla\mathcal{H})
			=
			\operatorname{tr}(\mathbf{J}\nabla^2\mathcal{H})
			=0 .
			\label{eq:canonical_divergence_zero}
		\end{equation}
		If $\mathbf J$ also satisfies the Jacobi identity, condition
        ~\eqref{eq:unimodular_condition} is the vanishing of its Poisson
        modular vector field relative to $\rho\,d\mathbf{x}$
        ~\cite{weinstein1997}. Without Jacobi it remains the sufficient
        divergence condition just proved; no Poisson-geometric
        interpretation is required.
	\end{remark}

	\begin{theorem}[Exact skewness under the Gibbs measure]
		\label{thm:gibbs}
		Consider the pendulum
		$\mathcal{H}(\theta,p)=\tfrac12p^2+1-\cos\theta$ on the
		cylinder $\mathcal{X}=\mathbb{S}^1\times\mathbb{R}$ with the
		constant interconnection
		$\mathbf{J}=\left[\begin{smallmatrix}0&1\\-1&0\end{smallmatrix}\right]$,
		the dictionary
		$\Psi=(\sin\theta,\,p,\,\cos\theta,\,\mathcal{H})^\intercal$,
		and the Gibbs measure
		$\nu_\beta=Z^{-1}e^{-\beta\mathcal{H}}\,d\theta\,dp$ with
		$\beta>0$. Then the conservative weak
		matrix~\eqref{eq:BJ_AJ} is exactly skew,
		$B_J=-B_J^\intercal$.
	\end{theorem}
	\begin{proof}
		The density $\rho=Z^{-1}e^{-\beta\mathcal{H}}=g(\mathcal{H})$
		is positive with constant $\mathbf{J}$, so
		$\operatorname{div}(\rho\mathbf{v}_J)=0$ by
		Lemma~\ref{lem:unimodular}(i). The measure factors as
		$e^{-\beta p^2/2}$ in $p$ times $e^{\beta\cos\theta}$ in
		$\theta$; the Gaussian factor makes the boundary flux at
		$p\to\pm\infty$ vanish for every product of dictionary
		functions, and $\theta$ is periodic. Both terms
		of~\eqref{eq:skew_defect} vanish, so $B_J=-B_J^\intercal$ by
		Lemma~\ref{lem:skew_defect}.
	\end{proof}
	The $(\sin\theta,p)$ entry gives an analytic cross check. With
	$\mathcal{K}_Jp=-\sin\theta$ and
	$\mathcal{K}_J\sin\theta=p\cos\theta$,
	\begin{equation}
		(B_J+B_J^\intercal)_{12}
		=
		\langle p^2\cos\theta\rangle_\beta
		-
		\langle\sin^2\theta\rangle_\beta ,
		\label{eq:gibbs_cross_check}
	\end{equation}
	which vanishes if and only if
	$\langle\sin^2\theta\rangle_\beta
	=\tfrac1\beta\langle\cos\theta\rangle_\beta$, using
	$\langle p^2\rangle_\beta=1/\beta$. For the von Mises marginal
	$\propto e^{\beta\cos\theta}$ this is the Bessel recurrence
	\begin{equation}
		I_0(\beta)-I_2(\beta)=\frac{2}{\beta}I_1(\beta),
		\label{eq:bessel_recurrence}
	\end{equation}
	so the identity holds for every $\beta>0$. The momentum moments
	entering the dissipative Dirichlet matrix are
	$m_1=\langle p\rangle_\beta=0$ and
	$m_2=\langle p^2\rangle_\beta=1/\beta$.

	Under the normalized uniform measure on
	$\mathbb{S}^1\times[-a,a]$ the same matrix is not skew. Here
	$\langle\cos\theta\rangle=0$ and
	$\langle\sin^2\theta\rangle=1/2$, so~\eqref{eq:gibbs_cross_check}
	equals $-1/2$ independently of $a$, while every other pair
	vanishes by parity. In the ordering
	$(\sin\theta,p,\cos\theta,\mathcal{H})$ the symmetric defect is
	\begin{equation}
		B_J+B_J^\intercal
		=
		-\tfrac12\bigl(e_1e_2^\intercal+e_2e_1^\intercal\bigr),
		\label{eq:uniform_defect}
	\end{equation}
	because the $(\sin\theta,p)$ integrand
	$\nabla(\sin\theta\,p)^\intercal\mathbf{v}_J
	=p^2\cos\theta-\sin^2\theta$ has mean $-1/2$. Experiment~E1
	verifies both the Gibbs exactness of Theorem~\ref{thm:gibbs} and
	this uniform-measure defect.

	\begin{remark}[Why the dissipative derivation matrix is not the Dirichlet matrix]
		The matrix
		$(C_R)_{ij}=\langle\psi_i,\mathcal{K}_R\psi_j\rangle_\nu$
		represents the first-order derivation
		$\mathcal{K}_R\psi_j=[\psi_j,\mathcal{H}]_R$. The matrix
		$(D_R)_{ij}=\int[\psi_i,\psi_j]_R\,d\nu$ represents the
		symmetric metric bracket. These are different objects. In
		general,
		\begin{equation}
			\langle\phi,\mathcal{K}_R\phi\rangle_\nu
			=
			\int_{\mathcal{X}}\phi[\phi,\mathcal{H}]_R\,d\nu
			\label{eq:CR_no_fixed_sign}
		\end{equation}
		has no fixed sign, whereas
		\begin{equation}
			\int_{\mathcal{X}}[\phi,\phi]_R\,d\nu\ge0 .
			\label{eq:DR_nonnegative_form}
		\end{equation}
		This distinction is the finite-dimensional counterpart of
		the one-dimensional counterexample in Section~\ref{sec:decomposition}.
	\end{remark}
	
	\begin{proposition}[Compatibility of the Galerkin drift with a fixed metric]
        \label{prop:galerkin_compatibility}
        Suppose the weak matrices have finite entries, $M\succ0$, and
        $B_J=-B_J^\intercal$. Write $P_M=M^{-1}$,
        $\operatorname{sym}C=(C+C^\intercal)/2$, and
        $\operatorname{skew}C=(C-C^\intercal)/2$. The coordinate-space
        Galerkin drift of the unforced plant is
        \begin{equation}
        A_G=(B_J-C_R)^\intercal M^{-1}
        =[-B_J+\operatorname{skew}C_R-\operatorname{sym}C_R]P_M.
        \label{eq:galerkin_compatibility}
        \end{equation}
        It admits $A_G=(S-D)P_M$ with $S=-S^\intercal$ and
        $D\succeq0$ if and only if $\operatorname{sym}C_R\succeq0$.
        In that case $S=-B_J+\operatorname{skew}C_R$ and
        $D=\operatorname{sym}C_R$. Positivity of $D_R$ alone does not
        imply this condition.
    \end{proposition}
    \begin{proof}
        The coefficient-space Galerkin matrix is $M^{-1}(B_J-C_R)$;
        transposition gives the coordinate drift. The symmetric and
        skew parts of $A_GM$ uniquely determine $-D$ and $S$.
    \end{proof}

    \begin{remark}[Constant--Hamiltonian obstruction and metric scope]
        \label{rem:constant_h_obstruction}
        If the dictionary span contains $1$ and $\mathcal H$, put
        $h=[\mathcal H,\mathcal H]_R$ and assume
        $\int h\,d\nu>0$ and all indicated moments exist. The symmetric
        dissipative weak form restricted to these two observables is
        \[
        \begin{bmatrix}
        0 & \tfrac12\int h\,d\nu\\
        \tfrac12\int h\,d\nu & \int \mathcal H h\,d\nu
        \end{bmatrix}.
        \]
        Its determinant is negative, since $\mathcal K_R1=0$ and
        $\mathcal K_R\mathcal H=h$. Thus $\operatorname{sym}C_R$ is
        indefinite for every admissible enlargement retaining this
        span and the same measure. Under the hypotheses of
        Proposition~\ref{prop:galerkin_compatibility}, exact Galerkin
        drift retention and dissipativity in $P_M$ are incompatible.
        Changing the storage metric is a different modeling choice and
        may remove the incompatibility. This obstruction does not rule
        out a quadratic physical-energy representation on the nonlinear
        lift image, nor does it explain a fitted storage obtained with
        a different empirical objective.
    \end{remark}

	\subsection{Projection Residuals and Closure Error}
	
	Along the original nonlinear pH dynamics, the lifted coordinate
	$z=\Psi(\mathbf{x})$ satisfies the exact identity
	\begin{equation}
		\dot z
		=
		D\Psi(\mathbf{x})
		\left[
		(\mathbf{J}(\mathbf{x})-\mathbf{R}(\mathbf{x}))
		\nabla\mathcal{H}(\mathbf{x})
		+
		\mathbf{G}(\mathbf{x})\mathbf{u}
		\right].
		\label{eq:exact_lifted_derivative}
	\end{equation}
	For a generic finite-dimensional linear lifted model, one can
	write
	\begin{equation}
		\dot z=A_Nz+B_N\mathbf{u}+r_N^A(\mathbf{x},\mathbf{u}),
		\label{eq:generic_projection_residual}
	\end{equation}
	where $r_N^A$ is the closure or projection residual. For the
	lifted pH surrogate introduced below, the corresponding
	residual is
	\begin{align}
		r_N(\mathbf{x},\mathbf{u})
		&\triangleq
		D\Psi(\mathbf{x})
		\left[
		(\mathbf{J}-\mathbf{R})\nabla\mathcal{H}
		+
		\mathbf{G}\mathbf{u}
		\right]
		\nonumber\\
		&\quad
		-
		\left[
		(S-D)P\Psi(\mathbf{x})+B\mathbf{u}
		\right].
		\label{eq:ph_projection_residual}
	\end{align}

	Here $r_N(\mathbf{x},\mathbf{u})$ is the difference between the
	exact lifted derivative and the surrogate prediction, with the
	sign convention exact minus surrogate.
	For the affine-input model, $r_N=0$ for every admissible input
    precisely when $D\Psi\,\mathbf v_0=A_N\Psi$ and
    $D\Psi\,\mathbf G=B_N$ on the considered set. The latter requires
    a constant lifted input map. Closure under input derivations instead
    leads to a bilinear model. Constants must be included when expressing
    affine closure as invariance of an observable span. A nonzero residual
    can include dictionary truncation, imposed-structure distortion,
    port approximation, regularization, and estimation error; it is not
    necessarily just the unconstrained Galerkin projection error.

	The weak matrices above should be interpreted accordingly.
	The matrices $M$, $B_J$, $C_R$, and $D_R$ are population weak
	objects defined with respect to $\nu$. The coefficient matrix
	$A_J=M^{-1}B_J$ gives the weak projection of the conservative
	derivation onto $\mathcal{V}_N$. Weak projection does not
	imply exact pointwise closure of the lifted dynamics. Hence
	exact algebraic identities such as
	$A_J^\intercal M+M A_J=0$ and $D_R\succeq0$ should be
	distinguished from the trajectory approximation accuracy of
	the surrogate.
	
	\subsection{Population Identities and Empirical Estimates}
	
	The matrices
	\[
	M,\qquad B_J,\qquad C_R,\qquad D_R
	\]
	are population quantities defined by integrals with respect to
	$\nu$. In computations, one often replaces these integrals by
	quadrature or empirical averages with
	\begin{equation}
		\widehat{\nu}_{N_{\rm samp}}
		=
		\frac{1}{N_{\rm samp}}
		\sum_{k=1}^{N_{\rm samp}}\delta_{\mathbf{x}_k},
		\label{eq:empirical_measure_samples}
	\end{equation}
	where $N_{\rm samp}$ denotes the sample count. The corresponding
	empirical inner product is
	\begin{equation}
		\langle f,g\rangle_{\widehat{\nu}_{N_{\rm samp}}}
		=
		\frac{1}{N_{\rm samp}}
		\sum_{k=1}^{N_{\rm samp}}f(\mathbf{x}_k)g(\mathbf{x}_k).
		\label{eq:empirical_inner_product}
	\end{equation}
	When $\nu$ is replaced by $\widehat{\nu}_{N_{\rm samp}}$, the
	same formulas produce finite-sample estimates
	\[
	\widehat M,\qquad
	\widehat B_J,\qquad
	\widehat C_R,\qquad
	\widehat D_R
	\]
	of the population matrices. In the rest of the paper, the meaning
	of the symbols is determined by the underlying inner product or
	measure. Finite data and noise
	introduce sampling and estimation error, so learned matrices
	should not be identified with exact population Galerkin
	matrices at finite sample size. An exact empirical Dirichlet
    matrix is already a sum of PSD Gram matrices and retains positivity.
    Conservative skewness need not hold empirically. Independent learned
    dissipation matrices require their own cone constraints; noisy
    reconstructions must be distinguished from exact Gram assembly.
	This empirical generator viewpoint is closely related to
	generator EDMD and data-driven Koopman-generator approximation
	methods~\cite{klus2020generator}.
	
	\subsection{Structure-Preserving Lifted pH Surrogates}
	
	The weak identities above motivate finite-dimensional models
	that preserve pH algebraic structure in lifted coordinates. A
	general lifted pH surrogate has the form
	\begin{equation}
		\dot z=(S-D)Pz+B\mathbf{u},
		\qquad
		\widetilde{\mathbf{y}}=B^\intercal Pz,
		\label{eq:kph_model}
	\end{equation}
	with
	\begin{equation}
		S=-S^\intercal,\qquad
		D=D^\intercal\succeq0,\qquad
		P=P^\intercal\succ0 .
		\label{eq:structure_constraints}
	\end{equation}
	The lifted storage is
	\begin{equation}
		\widetilde{\mathcal{H}}(z)
		=
		\frac12 z^\intercal Pz .
		\label{eq:lifted_storage}
	\end{equation}
	The output $\widetilde{\mathbf{y}}$ is the power-conjugate output
	associated with the lifted storage
	$\widetilde{\mathcal{H}}(z)=\frac12z^\intercal Pz$, mirroring
	$\mathbf{y}=\mathbf{G}^\intercal\nabla\mathcal{H}$ in the
	original pH system.
	This is an imposed structured model class. Proposition
    ~\ref{prop:galerkin_compatibility} gives the additional condition for
    the unforced Galerkin drift to belong to it with $P=P_M$. In the fit
    below, $P$ and $D$ need not equal $P_M$ and $D_R$. Exact agreement
    with the nonlinear lifted derivative requires $r_N=0$.
	For state-dependent port maps $\mathbf{G}(\mathbf{x})$, the
	constant matrix $B$ is a finite-dimensional approximation of the
	lifted input vector fields
	$D\Psi(\mathbf{x})\mathbf{G}(\mathbf{x})$. Exact representation
	requires $D\Psi(\mathbf{x})\mathbf{G}(\mathbf{x})=B$ on the
	considered set; otherwise their mismatch is included in $r_N(\mathbf{x},\mathbf{u})$; the exact bilinear input structure and a pendulum instance appear in Remark~\ref{rem:bilinear}.
	
	\begin{remark}[Exact bilinear input structure]
		\label{rem:bilinear}
		For control-affine dynamics
		$\dot{\mathbf{x}}=\mathbf{f}(\mathbf{x})
		+\sum_j\mathbf{g}_j(\mathbf{x})u_j$, the exact lifted input
		action is $D\Psi(\mathbf{x})\mathbf{G}(\mathbf{x})\mathbf{u}$,
		which admits a finite bilinear representation when each input
		derivation closes on the dictionary; otherwise an input-dependent
		residual remains:
		\begin{equation}
			\dot z=A_z z+\sum_j u_j N_j z+r_N ,
			\label{eq:bilinear_lift}
		\end{equation}
		where $N_j$ represents the input derivation
		$\mathcal{K}_{g_j}\psi=\nabla\psi^\intercal\mathbf{g}_j$ on the
		dictionary, the input analogue of the conservative
		matrix~\eqref{eq:BJ_AJ}. Bilinear Koopman models of this form
		are standard for control-affine
		systems~\cite{surana2016,goswami2017,peitz2020}. The
		constant-$B$ surrogate~\eqref{eq:kph_model} is one further
		approximation of this bilinear action, and its mismatch is part
		of $r_N$.

		For the pendulum, augment the dictionary with the constant
		observable,
		$\Psi=(1,\sin\theta,p,\cos\theta,\mathcal{H})^\intercal$,
		ordered $z_1=1$, $z_2=\sin\theta$, $z_3=p$, $z_4=\cos\theta$,
		$z_5=\mathcal{H}$. The single input channel is
		$\mathbf{g}=\mathbf{G}=(0,1)^\intercal$, and the dictionary is
		closed under $\partial_p$:
		\begin{equation}
			\mathcal{K}_g\Psi=\partial_p\Psi
			=(0,0,1,0,p)^\intercal
			=z_1 e_3+z_3 e_5 .
			\label{eq:pendulum_bilinear}
		\end{equation}
		The input channel is therefore exactly bilinear,
		$\mathcal{K}_g\Psi=N\Psi$, with the only nonzero entries
		$N_{3,1}=1$ and $N_{5,3}=1$. The best constant $B$ in
		$L^2(\nu_\beta)$ is the mean of $\partial_p\Psi$,
		$B^\star=(0,0,1,0,0)^\intercal$, with squared residual
		$\operatorname{Var}_{\nu_\beta}(p)=1/\beta$ from the
		$\partial_p\mathcal{H}=p$ entry. The port term is thus exactly
		representable in the bilinear class here, while the drift
		retains a closure residual because $p\cos\theta$ and
		$p\sin\theta$ leave the span. Adding the constant also closes
		$p^2=2\mathcal{H}-2+2\cos\theta$, reconciling with
		Section~\ref{sec:example_nonlinear}.
	\end{remark}

		\begin{remark}[Storage and output consistency]
		\label{rem:storage_output}
		Including $\mathcal H$ as an observable makes energy linear in
        ambient lifted coordinates, but does not preclude a PSD quadratic
        representation on the nonlinear lift image. For the augmented
        pendulum dictionary $(1,\sin\theta,p,\cos\theta,\mathcal H)$,
        \[
        P_H=(e_1-e_4)(e_1-e_4)^\intercal+e_2e_2^\intercal+e_3e_3^\intercal
        \succeq0
        \]
        gives the exact global identity
        \begin{equation}
        \tfrac12\Psi^\intercal P_H\Psi
        =\tfrac12[p^2+\sin^2\theta+(1-\cos\theta)^2]
        =\mathcal H.
        \label{eq:exact_energy_lift}
        \end{equation}
        $P_H$ is singular. A positive definite matrix cannot give zero
        energy at the nonzero equilibrium lift $(1,0,0,1,0)$.
        The smooth periodic lift $(\sin\theta,p,1-\cos\theta)$ instead
        represents $\mathcal H$ with $P=I_3$ and vanishes at the downward
        equilibrium. These facts concern energy representation, not
        closure of its dynamics or Galerkin compatibility in that metric.

        The numerical experiments fit storage by
        \begin{equation}
        \min_{P\succeq0}\;\sum_k\bigl(\tfrac12
        \Psi(\mathbf{x}_k)^\intercal P\Psi(\mathbf{x}_k)
        -\mathcal H(\mathbf{x}_k)\bigr)^2 .
        \label{eq:storage_sdp}
        \end{equation}
        This convex program does not ensure $P\succ0$. Results using
        $P^{-1}$ or positive definite storage require a separately
        checked lower eigenvalue bound, a restriction to a nonsingular
        subspace, or an explicitly imposed constraint $P\succeq\epsilon I$.
        The archived experiments did not impose the last constraint.

        For the five-observable dictionary, $y=p=z_3$. If $P\succ0$,
        choosing $B=P^{-1}e_3$ makes $B^\intercal P\Psi=p$ at the cost
        of changing the input fit. For singular $P$, the equation $PB=e_3$
        is solvable only if $e_3$ lies in its range. E2 compares a
        separately output-consistent fit; E3 independently refits the
        two-observable dictionary with free $B$ and does not impose exact
        physical-output matching.
    \end{remark}

	\subsection{Finite-Data Identification of the Lifted pH Surrogate}
	
	Given data
	\begin{equation}
		\{(\mathbf{x}_k,\mathbf{u}_k,\dot z_k)\}_{k=1}^{N_{\rm samp}},
		\qquad
		z_k=\Psi(\mathbf{x}_k),
		\qquad
		\dot z_k\approx D\Psi(\mathbf{x}_k)\dot{\mathbf{x}}_k,
		\label{eq:lifted_derivative_data}
	\end{equation}
	a structure-preserving continuous-time fit can be formulated as
	\begin{align}
		\min_{S,D,B}\quad
		&
		\sum_{k=1}^{N_{\rm samp}}
		\left\|
		\dot z_k-(S-D)Pz_k-B\mathbf{u}_k
		\right\|_2^2
		\nonumber\\
		&\quad
		+\lambda_S\|S\|_F^2
		+\lambda_D\|D\|_F^2
		\nonumber\\
		\mathrm{s.t.}\quad
		&
		S=-S^\intercal,\qquad
		D=D^\intercal\succeq0 .
		\label{eq:finite_data_ph_fit}
	\end{align}

	This constrained surrogate-fitting viewpoint is related to
	structure-preserving pH identification methods such as pH-DMD
	and non-intrusive pH realization from time-domain data
	\cite{morandin2023,cherifi2022}.
	One may either fix $P=P^\intercal\succ0$ from a chosen
	storage metric or learn it with additional regularization and
	constraints. The passivity statements below are conditional on
	the fitted matrices satisfying the pH constraints. In
	implementations, the constraints can be enforced by
	parameterizing
	\[
	S=\frac12(\Theta-\Theta^\intercal),
	\qquad
	D=LL^\intercal,
	\]
	for unconstrained matrices $\Theta$ and $L$.
	The regularization terms control the size of the fitted matrices
	but do not by themselves enforce pH structure. The constraints are
	enforced either directly or through parameterizations such as the
	ones above. If exact
	derivative data $\dot z_k$ are unavailable, they may be
	estimated from time-series data or replaced by a discrete-time
	residual. This introduces additional numerical error; the
	continuous-time identities in this paper are population and
	model identities, not guarantees that derivative estimation is
	exact.
	
		\begin{remark}[Convexity and identification]
		\label{rem:identification}
		With $P$ fixed, problem~\eqref{eq:finite_data_ph_fit} is a
		convex semidefinite-constrained least squares: the objective is
		convex quadratic in $(S,D,B)$, the skewness $S=-S^\intercal$ is
		linear, and $D\succeq0$ is a cone constraint. With $P$ learned
		the problem is bilinear in $(P,S,D,B)$; a cheap initialization
		fits an unconstrained drift $\widehat A$ and $B$ by linear
		regression, forms $\widehat AP^{-1}=S-D$ from the gauge
		relation $A_z=(S-D)P$, extracts $S$ and $D$ as the skew and
		negated symmetric parts of $\widehat AP^{-1}$ as
		in~\eqref{eq:linear_fixed_Q_decomposition} with $Q=P$, projects
		$D$ onto the positive semidefinite cone, and warm starts the
		constrained fit. To avoid derivative estimation, an
		integral-form regression fits $z(t_1)-z(t_0)$ against the time
		integral of the model by quadrature, as in generator EDMD and
		weak-form model selection~\cite{klus2020generator,messenger2021}.
		For iid samples from $\nu$ and square-integrable matrix-entry
        integrands, each empirical mean has root-mean-square error of
        order $N_{\rm samp}^{-1/2}$. High-probability bounds and
        trajectory-data rates require additional tail or dependence
        assumptions; none is asserted for the training trajectories here. The nonuniqueness of the underlying pH
		representation, and the gauge freedom that fixing $P$ removes,
		are discussed in Remark~\ref{rem:uniqueness_converse}.
	\end{remark}

	\begin{proposition}[Passivity of the lifted pH surrogate]
		\label{cor:passivity}
		If
		\[
		\begin{gathered}
			\dot z=(S-D)Pz+B\mathbf{u},\qquad
			S=-S^\intercal,\\
			D=D^\intercal\succeq0,\qquad
			P=P^\intercal\succ0,
		\end{gathered}
		\]
		and
		\[
		\widetilde{\mathcal{H}}(z)=\frac12 z^\intercal Pz,
		\qquad
		\widetilde{\mathbf{y}}=B^\intercal Pz,
		\]
		then
		\begin{equation}
			\dot{\widetilde{\mathcal{H}}}
			=
			-(Pz)^\intercal D(Pz)
			+
			\widetilde{\mathbf{y}}^\intercal\mathbf{u}
			\le
			\widetilde{\mathbf{y}}^\intercal\mathbf{u}.
			\label{eq:lifted_energy_balance}
		\end{equation}
	\end{proposition}
	
	\begin{proof}
		Differentiating the lifted storage along the surrogate gives
		\begin{align}
			\dot{\widetilde{\mathcal{H}}}
			&=
			z^\intercal P(S-D)Pz
			+
			z^\intercal PB\mathbf{u} \nonumber\\
			&=
			(Pz)^\intercal S(Pz)
			-
			(Pz)^\intercal D(Pz)
			+
			(B^\intercal Pz)^\intercal\mathbf{u}.
			\label{eq:lifted_passivity_proof}
		\end{align}
		Since $S=-S^\intercal$, $(Pz)^\intercal S(Pz)=0$. Since
		$D\succeq0$, $-(Pz)^\intercal D(Pz)\le0$. Finally,
		$(B^\intercal Pz)^\intercal\mathbf{u}
		=\widetilde{\mathbf{y}}^\intercal\mathbf{u}$, which proves
		\eqref{eq:lifted_energy_balance}.
	\end{proof}
	
	\begin{remark}[Projection residuals and scope of the surrogate]
		\label{rem:lifting_gap}
		Along the original nonlinear dynamics, suppose the exact
		lifted derivative is represented as
		\[
		\dot z=(S-D)Pz+B\mathbf{u}+r_N(\mathbf{x},\mathbf{u}).
		\]
		Then the lifted storage evolves as
		\begin{align}
			\dot{\widetilde{\mathcal{H}}}
			&=
			-(Pz)^\intercal D(Pz)
			+
			\widetilde{\mathbf{y}}^\intercal\mathbf{u}
			+
			(Pz)^\intercal r_N(\mathbf{x},\mathbf{u}).
			\label{eq:lifted_residual_balance}
		\end{align}
		Thus Proposition~\ref{cor:passivity} is an exact statement
		about the fitted or constructed finite-dimensional pH
		surrogate. Transferring passivity or stability back to the
		original nonlinear plant using this storage requires control
        of the signed residual term. Passivity at the surrogate output
        does not require lift injectivity; physical-output passivity
        additionally requires the port discrepancy to be accounted for.
        State convergence requires state separation or an appropriate
        invariant-set argument, not merely instantaneous output matching.
	\end{remark}
	
	The next section uses the lifted pH surrogate as the object
	of control design. All passivity and stability statements
	there are statements about the finite-dimensional lifted pH
	model unless additional residual and injectivity assumptions
	are explicitly invoked.

	\section{Passivity-Based Control for the Lifted pH Surrogate}
	\label{sec:control}
	
	The control design in this section is formulated for the
	finite-dimensional lifted pH surrogate introduced in
	Section~\ref{sec:approximation}:
	\begin{equation}
		\begin{gathered}
			\dot z=(S-D)Pz+B\mathbf{u},\\
			S=-S^\intercal,\qquad
			D=D^\intercal\succeq0,\qquad
			P=P^\intercal\succ0,
		\end{gathered}
		\label{eq:control_surrogate}
	\end{equation}
	where $z\in\mathbb{R}^N$, $S,D,P\in\mathbb{R}^{N\times N}$,
	and $B\in\mathbb{R}^{N\times m}$. The storage and output are
	\begin{equation}
		\widetilde{\mathcal{H}}(z)=\frac12 z^\intercal Pz,
		\qquad
		\widetilde{\mathbf{y}}=B^\intercal Pz .
		\label{eq:control_storage_output}
	\end{equation}
	All passivity and stability statements in this section are
	statements about the surrogate unless additional residual and
	injectivity assumptions are explicitly invoked.
	
	\subsection{Damping Injection}
	
	We consider the damping-injection law
	\begin{equation}
		\mathbf{u}
		=
		-K_d\widetilde{\mathbf{y}}
		+
		\mathbf{v},
		\qquad
		K_d=K_d^\intercal\succeq0,
		\label{eq:damping}
	\end{equation}
	where $\mathbf{v}$ is an external input or feedforward signal.
	When $K_d\succ0$, the injected port damping directly penalizes all
	nonzero components of the surrogate output
	$\widetilde{\mathbf{y}}$, although asymptotic stability still
	depends on the zero-dissipation detectability condition.
	
	\begin{proposition}[Damping injection for the lifted pH surrogate]
		\label{thm:stability}
		Consider the lifted pH surrogate~\eqref{eq:control_surrogate}
		with storage and output~\eqref{eq:control_storage_output}.
		Under the damping-injection law~\eqref{eq:damping}, the
		closed-loop storage satisfies
		\begin{equation}
			\dot{\widetilde{\mathcal{H}}}
			=
			-(Pz)^\intercal D(Pz)
			-
			\widetilde{\mathbf{y}}^\intercal K_d\widetilde{\mathbf{y}}
			+
			\widetilde{\mathbf{y}}^\intercal\mathbf{v}
			\le
			\widetilde{\mathbf{y}}^\intercal\mathbf{v}.
			\label{eq:damping_storage_balance}
		\end{equation}
		In particular, for $\mathbf{v}=0$, the surrogate storage is
		nonincreasing.
		
		Let
		\begin{equation}
			A_{\rm cl}
			\triangleq
			(S-D)P-BK_dB^\intercal P
			=
			(S-D-BK_dB^\intercal)P
			\label{eq:surrogate_acl}
		\end{equation}
		and define the zero-dissipation output matrix
		\begin{equation}
			C_{\rm diss}
			\triangleq
			\begin{bmatrix}
				D^{1/2}P\\
				K_d^{1/2}B^\intercal P
			\end{bmatrix}.
			\label{eq:zero_diss_output}
		\end{equation}
		Here $D^{1/2}$ and $K_d^{1/2}$ denote the symmetric positive
		semidefinite square roots of $D$ and $K_d$, respectively.
		If the pair $(A_{\rm cl},C_{\rm diss})$ is detectable, then
		the equilibrium $z=0$ of the surrogate with $\mathbf{v}=0$
		is asymptotically stable.
	\end{proposition}

	\begin{proof}
		Differentiating the storage along the surrogate gives
		\begin{align}
			\dot{\widetilde{\mathcal{H}}}
			&=
			z^\intercal P\dot z \nonumber\\
			&=
			z^\intercal P(S-D)Pz
			+
			z^\intercal PB\mathbf{u}.
			\label{eq:damping_proof_derivative}
		\end{align}
		Since $S=-S^\intercal$,
		\[
		z^\intercal PSPz=(Pz)^\intercal S(Pz)=0.
		\]
		Since $D\succeq0$,
		\[
		z^\intercal PDPz=(Pz)^\intercal D(Pz)\ge0.
		\]
		Finally,
		\[
		z^\intercal PB\mathbf{u}
		=
		(B^\intercal Pz)^\intercal\mathbf{u}
		=
		\widetilde{\mathbf{y}}^\intercal\mathbf{u}.
		\]
		Substituting
		$\mathbf{u}=-K_d\widetilde{\mathbf{y}}+\mathbf{v}$
		gives~\eqref{eq:damping_storage_balance}. For
		$\mathbf{v}=0$, the right-hand side is nonpositive, so the
		storage is nonincreasing.
		
		For $\mathbf{v}=0$, equality
		$\dot{\widetilde{\mathcal{H}}}=0$ holds exactly on
		$\ker C_{\rm diss}$. LaSalle's invariance principle for the
		finite-dimensional linear system $\dot z=A_{\rm cl}z$ implies
		that trajectories converge to the largest $A_{\rm cl}$-invariant
		subspace contained in $\ker C_{\rm diss}$, equivalently the
		unobservable subspace of $(A_{\rm cl},C_{\rm diss})$
		\cite{khalil2002nonlinear}. Detectability of this pair implies
		that all modes on this subspace are asymptotically stable. Hence
		$z(t)\to0$. Since $P\succ0$, the storage is positive definite, so
		the surrogate origin is asymptotically stable.
	\end{proof}
	
		The Lyapunov relation
	$A_{\rm cl}^\intercal P+PA_{\rm cl}=-2C_{\rm diss}^\intercal
	C_{\rm diss}$ with $P\succ0$ shows that detectability of
	$(A_{\rm cl},C_{\rm diss})$ is not only sufficient but also
	necessary for $A_{\rm cl}$ to be Hurwitz~\cite{khalil2002nonlinear}.

	\begin{remark}[Original-state conclusions]
		\label{rem:original_state_control}
		Proposition~\ref{thm:stability} concerns the independent
        surrogate state. For the actual lift $z=\Psi(\mathbf x)$, the
        additional term $(Pz)^\intercal r_N$ need not have a favorable
        sign. To infer physical state convergence from $z\to0$, require
        an explicit state-separation bound such as assumption (iii) in
        Theorem~\ref{thm:residual_bound}. Local dynamical observability
        alone does not imply this instantaneous bound.

        If a lift is shifted to vanish at an equilibrium, the model,
        feedback offsets, and residual must be recomputed consistently.
        Centering a constant observable produces an identically zero
        feature, which must be removed before requiring $M\succ0$.
        An exactly preserved constant coordinate also prevents the
        unreduced ambient drift from being Hurwitz.

        Physical and surrogate outputs must be distinguished. Implementing
        $-K_d B^\intercal P\Psi(\mathbf x)$ is the feedback analyzed here;
        replacing it with $-K_d\mathbf y$ changes the plant controller.
        Direct analysis with a physical storage is a separate route and
        does not require the fitted $P$ to be a plant Lyapunov function.
        Even a nonpositive signed residual work term does not transfer
        nominal detectability to the plant: the largest invariant subset
        of zero dissipation must be checked for the actual closed-loop
        vector field, or a strict plant Lyapunov inequality established.
        Full injectivity is stronger than target separation; for example,
        on $|\theta|\le2.5$ the E3 lift is not injective but satisfies
        $\|\Psi(\theta,p)\|\ge(\sin2.5/2.5)\|(\theta,p)\|$.
    \end{remark}

	\begin{theorem}[Residual-aware plant-level bound]
		\label{thm:residual_bound}
		Consider the locally Lipschitz closed-loop plant under the
        damping-injection law of Proposition~\ref{thm:stability} with
        $\mathbf v=0$, with equilibrium $\mathbf x_\star$. Assume:
		\emph{(i)} $(A_{\rm cl},C_{\rm diss})$ is detectable, so
		$A_{\rm cl}$ is Hurwitz and there exist $\kappa,\lambda>0$ with
		$\|e^{A_{\rm cl}t}\|\le\kappa e^{-\lambda t}$;
		\emph{(ii)} under $\mathbf{u}=-K_dB^\intercal P\Psi(\mathbf{x})$
		the plant trajectory stays in a compact set $\Omega$ on which
		$|r_N(\mathbf{x},\mathbf{u})|\le\bar\rho$;
		\emph{(iii)} $\Psi(\mathbf{x}_\star)=0$ and there is
		$\alpha\in\mathcal{K}_\infty$ with
		$\alpha(|\mathbf{x}-\mathbf{x}_\star|)\le|\Psi(\mathbf{x})|$ on
		$\Omega$. Then the exact lifted trajectory
		$z(t)=\Psi(\mathbf{x}(t))$ satisfies
		$\dot z=A_{\rm cl}z+r_N(\mathbf{x}(t),\mathbf{u}(t))$ and
		\begin{equation}
			\limsup_{t\to\infty}|z(t)|\le\frac{\kappa\bar\rho}{\lambda},
			\quad
			\limsup_{t\to\infty}|\mathbf{x}(t)-\mathbf{x}_\star|
			\le\alpha^{-1}\!\Bigl(\frac{\kappa\bar\rho}{\lambda}\Bigr).
			\label{eq:residual_bound}
		\end{equation}
	\end{theorem}
	\begin{proof}
		Substituting $\mathbf{u}=-K_dB^\intercal Pz$ into the exact
		lifted derivative $\dot z=(S-D)Pz+B\mathbf{u}+r_N$
		of~\eqref{eq:lifted_residual_balance} gives
		$\dot z=A_{\rm cl}z+r_N$ with $A_{\rm cl}$ as
		in~\eqref{eq:surrogate_acl}. Variation of constants yields
		$z(t)=e^{A_{\rm cl}t}z(0)+\int_0^t e^{A_{\rm cl}(t-s)}r_N(s)\,ds$,
		so by (i) and (ii)
		\[
		|z(t)|\le\kappa e^{-\lambda t}|z(0)|
		+\kappa\bar\rho\,\frac{1-e^{-\lambda t}}{\lambda},
		\]
		hence $\limsup_{t\to\infty}|z(t)|\le\kappa\bar\rho/\lambda$. By
		(iii), $\alpha(|\mathbf{x}-\mathbf{x}_\star|)\le|\Psi(\mathbf{x})|
		=|z|$ and $\alpha^{-1}$ is nondecreasing, giving the second
		bound.
	\end{proof}
	\begin{remark}[Computable residual bound]
		\label{rem:computable_residual}
		A maximum over finitely many residual samples is a lower bound
        on the supremum, not a uniform upper bound. If the closed-loop
        residual is $L$-Lipschitz on $\Omega$ and samples form an
        $h$-net, then $\max_k\|r_N(x_k,u(x_k))\|+Lh$ is a valid upper
        bound, provided sample-evaluation errors are also enclosed.
        Interval or analytic bounds are alternatives. The theorem also
        requires trajectory confinement, which a residual grid does not
        establish. The archived E3 grid diagnostic alone therefore does
        not instantiate this theorem. The variation-of-constants
        argument is standard; proportional residual alternatives appear
        in \cite{mamakoukas2020stable,straesser2026safedmd}.
    \end{remark}

	\begin{remark}[Role of the storage matrix \(P\)]
		The control identities above use a general storage matrix
		$P=P^\intercal\succ0$. The output paired with this storage is
		$\widetilde{\mathbf{y}}=B^\intercal Pz$. For $P=I$, this
		reduces to the Euclidean lifted storage and output
		$B^\intercal z$. For general $P$, the energy identity must use
		the output consistent with the chosen storage metric. This is
		why the damping-injection law~\eqref{eq:damping} uses
		$\widetilde{\mathbf{y}}$, not an output defined independently
		of $P$.
	\end{remark}
	
	\subsection{Reference Tracking and Error Coordinates}
	
	The damping-injection proposition is stated at the origin. For a constant target
	$z_\star$, define $e=z-z_\star$ only when there exists a steady
	input $\mathbf{u}_\star$ satisfying
	$0=(S-D)Pz_\star+B\mathbf{u}_\star$. The same calculation applies
	to $e$ if the shifted dynamics retain the same pH form. Otherwise
	tracking requires additional feasibility and local-model
	assumptions.
	
	\subsection{MPC Formulation for the Lifted Surrogate}
	\label{sec:mpc}

	Because the lifted pH surrogate is finite-dimensional and
	control-affine, it can serve as the prediction model in an MPC
	controller for the surrogate, as in Koopman prediction-based MPC
    \cite{korda2018linear}. A discrete-time model
	\begin{equation}
		z_{k+1}=A_dz_k+B_d\mathbf{u}_k
		\label{eq:discrete_surrogate}
	\end{equation}
	is obtained by a passivity-preserving discretization, for which
	Proposition~\ref{prop:midpoint} gives an exact discrete energy
	balance. A typical finite-horizon problem is
	\begin{equation}
		\begin{aligned}
			\min_{\{\mathbf{u}_k\}_{k=0}^{N_h-1}}\quad
			&
			\sum_{k=0}^{N_h-1}
			\left(
			\|z_k-z_{\rm ref}\|_{Q_{\rm mpc}}^2
			+
			\|\mathbf{u}_k\|_{R_{\rm mpc}}^2
			\right)
			\\
			&\quad
			+
			\|z_{N_h}-z_{\rm ref}\|_{P_f}^2 \\
			\mathrm{s.t.}\quad
			&
			z_{k+1}=A_dz_k+B_d\mathbf{u}_k,
		\end{aligned}
		\label{eq:mpc}
	\end{equation}
	with any state or input constraints, where $Q_{\rm mpc}\succeq0$
	and $R_{\rm mpc}\succ0$ are design weights distinct from the pH
	storage and dissipation matrices and $P_f\succeq0$ is a terminal
	cost. Recursive feasibility and stability follow from standard
	terminal ingredients for $(A_d,B_d)$~\cite{mayne2000constrained}, not
	from $D\succeq0$ alone, and $P_f$ need not equal $P$. No MPC
	theorem is claimed here.

	\begin{proposition}[Exact discrete passivity of the implicit midpoint scheme]
		\label{prop:midpoint}
		The implicit midpoint scheme
		\begin{equation}
			z_+=z+\Delta t\bigl[(S-D)Pz_m+B\mathbf{u}\bigr],
			\qquad
			z_m=\tfrac12(z+z_+),
			\label{eq:midpoint_scheme}
		\end{equation}
		is well defined for every $\Delta t>0$ and satisfies the exact
		balance
		\begin{equation}
			\begin{aligned}
				\widetilde{\mathcal{H}}(z_+)-\widetilde{\mathcal{H}}(z)
				&=\Delta t\bigl[-(Pz_m)^\intercal D(Pz_m)
				+\widetilde{\mathbf{y}}_m^\intercal\mathbf{u}\bigr],\\
				\widetilde{\mathbf{y}}_m&=B^\intercal Pz_m .
			\end{aligned}
			\label{eq:midpoint_balance}
        \end{equation}
        The supply uses the midpoint output $\widetilde{\mathbf y}_m$.
        Feedback evaluated at the previous sample is a different
        discretization and does not inherit the same damping identity.
    \end{proposition}
	\begin{proof}
		Since $P=P^\intercal$,
		\[
		\begin{aligned}
		\widetilde{\mathcal{H}}(z_+)-\widetilde{\mathcal{H}}(z)
		&=\tfrac12(z_++z)^\intercal P(z_+-z)\\
		&=\Delta t\,z_m^\intercal P[(S-D)Pz_m+B\mathbf{u}],
		\end{aligned}
		\]
		using
		$z_+-z=\Delta t[(S-D)Pz_m+B\mathbf{u}]$ and $z_++z=2z_m$. The
		term $(Pz_m)^\intercal S(Pz_m)=0$ by $S=-S^\intercal$ and
		$-(Pz_m)^\intercal D(Pz_m)\le0$ by $D\succeq0$, which
		gives~\eqref{eq:midpoint_balance}. Eliminating $z_m$ yields
		$(I-\tfrac{\Delta t}{2}A)z_+=(I+\tfrac{\Delta t}{2}A)z
		+\Delta t B\mathbf{u}$ with $A=(S-D)P$. Now $A$ is similar to
		$P^{1/2}(S-D)P^{1/2}$, whose symmetric part is
		$-P^{1/2}DP^{1/2}\preceq0$, so every eigenvalue of $A$ has
		nonpositive real part and $I-\tfrac{\Delta t}{2}A$ is invertible
		for every $\Delta t>0$. Forward Euler, by contrast, matches the
		balance with the additional term
        $\tfrac12\Delta t^2 f^\intercal P f\ge0$, where
        $f=(S-D)Pz+B\mathbf u$. This nonnegative term can destroy
        discrete dissipation. This
		is the linear pH instance of the discrete-gradient and midpoint
		energy identities~\cite{gonzalez1996,mclachlan1999}.
	\end{proof}

	\section{Illustrative Examples}
	\label{sec:examples}
	\label{sec:example}
	
	The examples below illustrate the bracket decomposition. They distinguish exact pH vector-field
	identities, Galerkin weak-form and Dirichlet matrices,
	finite-dimensional surrogate models, and closure residuals.
	They locate dissipative positivity in the metric bracket, as the one-dimensional counterexample in Section~\ref{sec:decomposition} makes precise.
	
	\subsection{Linear pH Systems: Exact Lifted pH Representation}
	\label{sec:example_linear}
	
	Consider the linear pH system
	\begin{equation}
		\dot{\mathbf{x}}
		=
		(\mathbf{J}-\mathbf{R})Q\mathbf{x}
		+
		\mathbf{G}\mathbf{u},
		\qquad
		\mathcal{H}(\mathbf{x})
		=
		\frac12\mathbf{x}^\intercal Q\mathbf{x},
		\label{eq:linear_ph}
	\end{equation}
	where
	\[
	Q=Q^\intercal\succ0,\qquad
	\mathbf{J}^\intercal=-\mathbf{J},\qquad
	\mathbf{R}=\mathbf{R}^\intercal\succeq0.
	\]
	Then $\nabla\mathcal{H}=Q\mathbf{x}$, and the bracket
	identities give
	\begin{equation}
		\{\mathcal{H},\mathcal{H}\}_J
		=
		(Q\mathbf{x})^\intercal\mathbf{J}Q\mathbf{x}
		=
		0,
		\label{eq:linear_interconnection_zero}
	\end{equation}
	while
	\begin{equation}
		[\mathcal{H},\mathcal{H}]_R
		=
		(Q\mathbf{x})^\intercal\mathbf{R}Q\mathbf{x}
		\ge0.
		\label{eq:linear_metric_hamiltonian}
	\end{equation}
	Consequently, for the unforced drift,
	\begin{equation}
		\mathcal{K}_0\mathcal{H}
		=
		-[\mathcal{H},\mathcal{H}]_R
		=
		-(Q\mathbf{x})^\intercal\mathbf{R}Q\mathbf{x}
		\le0.
		\label{eq:linear_energy_unforced}
	\end{equation}
	With input,
	\begin{equation}
		\mathcal{K}_{\mathbf{u}}\mathcal{H}
		=
		-[\mathcal{H},\mathcal{H}]_R
		+
		\mathbf{y}^\intercal\mathbf{u},
		\qquad
		\mathbf{y}=\mathbf{G}^\intercal Q\mathbf{x}.
		\label{eq:linear_energy_forced}
	\end{equation}
	
	\begin{proposition}[Linear pH systems have zero closure residual in state coordinates]
		\label{prop:linear_exact}
		For the linear pH system~\eqref{eq:linear_ph} and dictionary
		$z=\Psi(\mathbf{x})=\mathbf{x}$, the lifted pH surrogate with
		\begin{equation}
			S=\mathbf{J},\qquad
			D=\mathbf{R},\qquad
			P=Q,\qquad
			B=\mathbf{G}
			\label{eq:linear_recovery}
		\end{equation}
		reproduces the original dynamics exactly and has zero closure
		residual.
	\end{proposition}
	
	\begin{proof}
		With $z=\mathbf{x}$, the lifted derivative is
		\begin{equation}
			\dot z
			=
			(\mathbf{J}-\mathbf{R})Qz+\mathbf{G}\mathbf{u}
			=
			(S-D)Pz+B\mathbf{u}.
			\label{eq:linear_exact_surrogate}
		\end{equation}
		This is exactly the original vector field in lifted
		coordinates, so the residual $r_N$ from
		Section~\ref{sec:approximation} vanishes.
		For the identity dictionary $z=\Psi(\mathbf{x})=\mathbf{x}$,
		the dictionary subspace
		$\mathcal{V}_N=\operatorname{span}\{x_1,\ldots,x_n\}$ is
		invariant under the unforced linear generator because
		$\mathcal{K}_0x_i=((\mathbf{J}-\mathbf{R})Q\mathbf{x})_i$ is
		linear in $\mathbf{x}$. The lifted input term is also represented
		exactly because $D\Psi(\mathbf{x})\mathbf{G}=\mathbf{G}$. Hence
		the closure residual vanishes.
	\end{proof}
	
	For the state-coordinate choice $z=\mathbf{x}$, the lifted storage
	is
	\begin{equation}
		\widetilde{\mathcal{H}}(z)
		=
		\frac12z^\intercal Pz
		=
		\frac12\mathbf{x}^\intercal Q\mathbf{x}
		=
		\mathcal{H}(\mathbf{x}).
		\label{eq:linear_storage_match}
	\end{equation}
	Thus the linear example gives an exact pH representation of
	the vector field in lifted coordinates. It is not a proof that
	the Galerkin weak-form matrix of the first-order derivation
	$\mathcal{K}_R$ has any sign guarantee.
	
	The distinction is visible even for the identity dictionary
	$\psi_i(\mathbf{x})=x_i$. If $\nu$ is a probability measure
	and $\mathbf{R}$ is constant, then
	\begin{equation}
		(D_R)_{ij}
		=
		\int
		\nabla x_i^\intercal\mathbf{R}\nabla x_j\,d\nu
		=
		R_{ij}.
		\label{eq:linear_dirichlet_identity}
	\end{equation}
	By contrast,
	\begin{equation}
		(C_R)_{ij}
		=
		\int x_i[x_j,\mathcal{H}]_R\,d\nu
		=
		\int x_i e_j^\intercal\mathbf{R}Q\mathbf{x}\,d\nu
		\label{eq:linear_CR_identity}
	\end{equation}
	depends on the second moments of $\nu$ and is generally not
	the same object as $D_R$. The positive matrix in the Galerkin
	bracket sense is $D_R$, while the exact state-space pH
	dissipation matrix in the linear vector field is
	$\mathbf{R}$.
	
	For the effort-coordinate choice
	\[
	z=Q\mathbf{x}=\nabla\mathcal{H},
	\]
	then
	\[
	\dot z=Q(\mathbf{J}-\mathbf{R})z+Q\mathbf{G}\mathbf{u}.
	\]
	This can still be written in lifted pH form with
	\[
	P=Q^{-1},\qquad
	S=Q\mathbf{J}Q,\qquad
	D=Q\mathbf{R}Q,\qquad
	B=Q\mathbf{G},
	\]
	because
	\[
	(S-D)Pz
	=
	(Q\mathbf{J}Q-Q\mathbf{R}Q)Q^{-1}z
	=
	Q(\mathbf{J}-\mathbf{R})z.
	\]
	Since $Q\mathbf{J}Q$ is skew-symmetric and
	$Q\mathbf{R}Q$ is positive semidefinite, the pH structure is
	preserved under this coordinate change. Exact representation
	therefore depends on both the chosen coordinates and the
	storage metric $P$.
	
	\subsection{Damped Nonlinear Pendulum: Brackets and Closure Residual}
	\label{sec:example_nonlinear}
	
	Let $\mathbf{x}=(\theta,p)^\intercal$ and
	\[
	\mathcal{H}(\theta,p)=\frac12p^2+1-\cos\theta.
	\]
	Consider the damped pendulum with
	\begin{equation}
		\mathbf{J}
		=
		\begin{bmatrix}
			0 & 1\\
			-1 & 0
		\end{bmatrix},
		\quad
		\mathbf{R}
		=
		\begin{bmatrix}
			0 & 0\\
			0 & b
		\end{bmatrix},
		\quad
		\mathbf{G}
		=
		\begin{bmatrix}
			0\\
			1
		\end{bmatrix},
		\quad b>0.
		\label{eq:pendulum_ph}
	\end{equation}
	Since
	\begin{equation}
		\nabla\mathcal{H}
		=
		\begin{bmatrix}
			\sin\theta\\
			p
		\end{bmatrix},
		\label{eq:pendulum_fields}
	\end{equation}
	the dynamics are
	\begin{equation}
		\dot\theta=p,
		\qquad
		\dot p=-\sin\theta-bp+u,
		\qquad
		y=p.
		\label{eq:pendulum_dynamics}
	\end{equation}
	The Hamiltonian brackets are
	\begin{equation}
		\{\mathcal{H},\mathcal{H}\}_J=0,
		\qquad
		[\mathcal{H},\mathcal{H}]_R=bp^2.
		\label{eq:pendulum_energy_brackets}
	\end{equation}
	Therefore
	\begin{equation}
		\mathcal{K}_{\mathbf{u}}\mathcal{H}
		=
		-bp^2+y^\intercal u
		=
		-bp^2+pu.
		\label{eq:pendulum_energy_balance}
	\end{equation}
	This is the pH energy identity for the Hamiltonian
	observable. It does not assert that $\mathcal{K}_R$ is
	positive on all observables.
	
	Now choose
	\begin{equation}
		\Psi(\theta,p)
		=
		\begin{bmatrix}
			\sin\theta\\
			p\\
			\cos\theta\\
			\mathcal{H}(\theta,p)
		\end{bmatrix}.
		\label{eq:pendulum_dictionary}
	\end{equation}
	This dictionary is useful for illustrating generator actions,
	but it is not closed under the generator. Direct calculation
	gives
	\begin{align}
		\mathcal{K}_J\sin\theta &= p\cos\theta,
		& \mathcal{K}_R\sin\theta &= 0, \nonumber\\
		\mathcal{K}_J p &= -\sin\theta,
		& \mathcal{K}_R p &= bp, \nonumber\\
		\mathcal{K}_J\cos\theta &= -p\sin\theta,
		& \mathcal{K}_R\cos\theta &= 0, \nonumber\\
		\mathcal{K}_J\mathcal{H} &= 0,
		& \mathcal{K}_R\mathcal{H} &= bp^2.
		\label{eq:pendulum_gen}
	\end{align}
	The identities involving $\mathcal{H}$ reproduce
	\eqref{eq:pendulum_energy_balance}. The remaining identities
	are generator identities for specific observables. The fact
	that $\mathcal{K}_R\mathcal{H}=bp^2\ge0$ is special to the
	Hamiltonian observable and does not imply positivity of
	$\mathcal{K}_R$ on arbitrary observables.
	
	The dictionary~\eqref{eq:pendulum_dictionary} is not closed:
	$p\cos\theta$ and $p\sin\theta$ are not generally in
	$\operatorname{span}\{\sin\theta,p,\cos\theta,\mathcal{H}\}$.
	Because the constant observable is not included,
	$p^2=2\mathcal{H}-2+2\cos\theta$ is also not in the listed span.
	Adding $1$ would close this particular term, although
	$p\cos\theta$ and $p\sin\theta$ would still remain outside the span. The augmented dictionary and its exact bilinear input channel appear in Remark~\ref{rem:bilinear}.
	Thus the lifted derivative
	$\dot z=D\Psi(\theta,p)\dot{\mathbf{x}}$ cannot generally be
	represented exactly by a linear finite-dimensional model in
	this dictionary. The residual $r_N$ introduced in
	Section~\ref{sec:approximation} is nonzero in general. This
	illustrates the difference between exact bracket identities
	and exact finite-dimensional closure. The weak identities of
	Proposition~\ref{prop:splitting} remain exact at the
	population level for the chosen test functions and measure,
	but trajectory closure of the finite-dimensional dictionary is
	a separate approximation question.
	
	Let
	\[
	\psi_1=\sin\theta,\qquad
	\psi_2=p,\qquad
	\psi_3=\cos\theta,\qquad
	\psi_4=\mathcal{H}.
	\]
	Since
	\[
	[\psi_i,\psi_j]_R
	=
	b\,\partial_p\psi_i\,\partial_p\psi_j,
	\qquad
	\partial_p\Psi
	=
	\begin{bmatrix}
		0\\
		1\\
		0\\
		p
	\end{bmatrix},
	\]
	the Dirichlet matrix is
	\begin{equation}
		D_R
		=
		b\int
		\begin{bmatrix}
			0\\
			1\\
			0\\
			p
		\end{bmatrix}
		\begin{bmatrix}
			0 & 1 & 0 & p
		\end{bmatrix}
		d\nu .
		\label{eq:pendulum_DR_integral}
	\end{equation}
	On the cylinder $S^1\times\mathbb{R}$, or on a compact sampled
	region with finite moments, if $\nu$ is a probability measure this
	becomes
	\begin{equation}
		D_R
		=
		b
		\begin{bmatrix}
			0 & 0 & 0 & 0\\
			0 & 1 & 0 & m_1\\
			0 & 0 & 0 & 0\\
			0 & m_1 & 0 & m_2
		\end{bmatrix}.
		\label{eq:pendulum_DR_matrix}
	\end{equation}
	Here
	\[
	m_1=\int p\,d\nu,\qquad
	m_2=\int p^2\,d\nu.
	\]
	Under any measure symmetric in $p$, $m_1=0$; for the Gibbs measure
	$\nu_\beta$, $m_1=0$ and $m_2=1/\beta$.
	For any $c\in\mathbb{R}^4$,
	\begin{equation}
		c^\intercal D_R c
		=
		b\int(c_2+c_4p)^2\,d\nu
		\ge0.
		\label{eq:pendulum_DR_psd}
	\end{equation}
	This is the finite-dimensional dissipative object that carries a sign:
	$D_R$ is positive semidefinite because it represents the
	metric bracket.
	
	The first-order dissipative derivation matrix is different:
	\begin{equation}
		(C_R)_{ij}
		=
		\int\psi_i[\psi_j,\mathcal{H}]_R\,d\nu.
		\label{eq:pendulum_CR_definition}
	\end{equation}
	For the pendulum,
	\begin{align*}
		[\psi_1,\mathcal{H}]_R&=0,
		&[\psi_2,\mathcal{H}]_R&=bp,\\
		[\psi_3,\mathcal{H}]_R&=0,
		&[\psi_4,\mathcal{H}]_R&=bp^2.
	\end{align*}
	Thus only the second and fourth columns are generally
	nonzero:
	\begin{equation}
		(C_R)_{i2}
		=
		b\int\psi_i p\,d\nu,
		\qquad
		(C_R)_{i4}
		=
		b\int\psi_i p^2\,d\nu.
		\label{eq:pendulum_CR_columns}
	\end{equation}
	The matrix $C_R$ depends on mixed moments of the sampling
	measure. It is generally not equal to $D_R$, and it is
	generally neither symmetric nor positive semidefinite. The
	finite-dimensional positive metric-bracket matrix is $D_R$,
	not $C_R$.
	
	Given sampled data, the empirical Dirichlet matrix can be
	estimated by
	\begin{equation}
		\widehat D_R
		=
		\frac{b}{N_{\rm samp}}
		\sum_{k=1}^{N_{\rm samp}}
		\begin{bmatrix}
			0\\
			1\\
			0\\
			p_k
		\end{bmatrix}
		\begin{bmatrix}
			0 & 1 & 0 & p_k
		\end{bmatrix}.
		\label{eq:pendulum_empirical_DR}
	\end{equation}
	This empirical matrix is positive semidefinite by
	construction. A lifted pH surrogate can be fit using the
	constrained least-squares formulation of
	Section~\ref{sec:approximation}. The residual norm
	\[
	\sum_k
	\left\|
	\dot z_k-(S-D)Pz_k-Bu_k
	\right\|^2
	\]
	measures closure and modeling error for the chosen
	dictionary. Finite data are used only to form approximate
	estimates for this chosen model class.
	
		\subsection{Rigid Body: State-Dependent Interconnection}
	\label{sec:example_rigidbody}

	The pendulum uses a constant $\mathbf{J}$. The rigid body
	illustrates a state-dependent interconnection, for which
	skew-symmetry of $\mathbf{J}(\mathbf{x})$ alone does not settle
	skewness of the weak matrix, and the measure conditions of
	Section~\ref{sec:approximation} carry the argument. On
	$\mathbb{R}^3$, take the angular momentum $\mathbf{x}$, the energy
	\begin{equation}
		\mathcal{H}(\mathbf{x})=\tfrac12\mathbf{x}^\intercal
		\mathbf{I}_b^{-1}\mathbf{x},
		\qquad
		\nabla\mathcal{H}=\mathbf{I}_b^{-1}\mathbf{x},
		\label{eq:rb_hamiltonian}
	\end{equation}
	and the state-dependent interconnection
	$\mathbf{J}(\mathbf{x})=\widehat{\mathbf{x}}$, the cross-product
	matrix with $\widehat{\mathbf{x}}\mathbf{y}=\mathbf{x}\times
	\mathbf{y}$. The controlled dynamics are the damped Euler
	equations
	\begin{equation}
		\dot{\mathbf{x}}
		=
		\widehat{\mathbf{x}}\nabla\mathcal{H}
		-\mathbf{R}\nabla\mathcal{H}
		+\mathbf{G}\mathbf{u}
		=
		\mathbf{x}\times(\mathbf{I}_b^{-1}\mathbf{x})
		-\mathbf{R}\mathbf{I}_b^{-1}\mathbf{x}
		+\mathbf{G}\mathbf{u},
		\label{eq:rb_dynamics}
	\end{equation}
	with conservative flow
	$\mathbf{v}_J=\widehat{\mathbf{x}}\nabla\mathcal{H}
	=\mathbf{x}\times(\mathbf{I}_b^{-1}\mathbf{x})$. The instance
	studied in Experiment~E4 uses
	$\mathbf{I}_b=\operatorname{diag}(1,2,3)$,
	$\mathbf{R}=0.05\,\mathbf{I}_3$, and $\mathbf{G}=\mathbf{e}_3$,
	from the numerics configuration. Three structural facts follow
	from the $\mathfrak{so}(3)$ interconnection alone, independently
	of $\mathbf{I}_b$, $\mathbf{R}$, and $\mathbf{G}$.

	\emph{(a) Unimodularity.} For each $k$,
	\begin{equation}
		\sum_i\partial_i\widehat{\mathbf{x}}_{ik}=0,
		\label{eq:rb_unimodular}
	\end{equation}
	because each nonzero entry of $\widehat{\mathbf{x}}$ is a
	coordinate differentiated with respect to a different coordinate.
	Hence Lebesgue measure meets the unimodularity
	condition~\eqref{eq:unimodular_condition} of
	Lemma~\ref{lem:unimodular}: $\mathfrak{so}(3)$ is unimodular, and
	$\operatorname{div}(\mathbf{v}_J)=0$ for every $\mathcal{H}$.

	\emph{(b) Tangency to spheres.} Since
	$\mathbf{v}_J\cdot\mathbf{x}
	=\mathbf{x}\cdot(\mathbf{x}\times\nabla\mathcal{H})=0$, the
	conservative flow is tangent to every sphere
	$\{|\mathbf{x}|=r\}$. On any ball of radius $r$ with a radial
	density $\rho(|\mathbf{x}|)$, $\nabla\rho$ is parallel to
	$\mathbf{x}$, so $\nabla\rho\cdot\mathbf{v}_J=0$ and, with
	\eqref{eq:rb_unimodular},
	$\operatorname{div}(\rho\mathbf{v}_J)=0$; on the boundary the flux
	integrand carries
	$\mathbf{v}_J\cdot\mathbf{n}=\mathbf{v}_J\cdot\mathbf{x}/
	|\mathbf{x}|=0$. Both terms of the defect
	identity~\eqref{eq:skew_defect} vanish, so $B_J$ is exactly skew
	for every dictionary, by Lemma~\ref{lem:skew_defect}.

	\emph{(c) Casimir in the kernel.} The function
	$C(\mathbf{x})=\tfrac12|\mathbf{x}|^2$ is a Casimir,
	\begin{equation}
		\mathcal{K}_J C
		=
		\nabla C^\intercal\mathbf{J}\nabla\mathcal{H}
		=
		\mathbf{x}^\intercal(\mathbf{x}\times\nabla\mathcal{H})
		=0
		\label{eq:rb_casimir}
	\end{equation}
	for every $\mathcal{H}$. Hence, when $C$ lies in the dictionary
	span with coefficient vector $c_C$, so that
	$\sum_i(c_C)_i\psi_i=C$, the conservative weak matrix annihilates
	it,
	\begin{equation}
		B_J c_C=0 ,
		\label{eq:rb_casimir_kernel}
	\end{equation}
	because $(B_J c_C)_i
	=\int_{\mathcal{X}}\psi_i\,\mathcal{K}_J C\,d\nu=0$.

	Experiment~E4 computes $B_J$ in closed form for the monomial
	dictionary of degree at most two on the unit ball, using exact
	ball moments. It verifies skewness and the Casimir kernel
	property~\eqref{eq:rb_casimir_kernel} symbolically, and measures
	the $O(N^{-1/2})$ defect of Monte Carlo estimates of $B_J$. The
	damped rigid body is then fit with a structure-preserving
	surrogate, as for the pendulum.

	\subsection{What the Examples Demonstrate}
	
	The linear pH example shows that with state coordinates and
	the storage metric $P=Q$, the lifted pH surrogate is
	exact and has zero residual. It also shows that exact pH
	representation is a vector-field statement, separate from positivity of the weak matrix $C_R$. The
	pendulum example shows that the Hamiltonian identity
	\[
	\mathcal{K}_{\mathbf{u}}\mathcal{H}
	=
	-[\mathcal{H},\mathcal{H}]_R+y^\intercal u
	\]
	holds exactly, while finite dictionaries can still have
	nonzero closure residuals. The Dirichlet matrix $D_R$ remains
	positive semidefinite because it represents the metric
	bracket, whereas $C_R$ generally has no sign guarantee. The rigid-body example adds the state-dependent case: skew-symmetry of $\mathbf{J}(\mathbf{x})$ is not enough, and Lebesgue measure yields exact weak skewness only because $\mathfrak{so}(3)$ is unimodular and the conservative flow is tangent to spheres. These examples support the paper's thesis: preserve the bracket, Dirichlet, and pH surrogate structure, and locate dissipative positivity in the metric bracket rather than in the first-order derivation or in exact finite-dimensional Koopman closure.
	
	\subsection{Numerical Experiments}
	\label{sec:numerics}

	The archived numerical values below use generated macros.
    Reproduction also requires constants and seeds defined in the scripts,
    not only the configuration file. E1 uses
    $(\sin\theta,p,\cos\theta,\mathcal H)$; E2 uses
    $(1,\sin\theta,p,\cos\theta,\mathcal H)$ and 30 training/10 test
    trajectories; E3 independently fits $(\sin\theta,p)$ using all
    40 trajectories from its training generator. Consequently E3 does
    not implement the five-dimensional model plotted in E2.
    Additional E3 diagnostics below are rounded from the separately
    exported reproduced model; archived figures and data are unchanged.

	\emph{E1 (weak form).} Under the Gibbs measure $\nu_\beta$ the conservative
	weak matrix is skew to machine precision,
	$\max|B_J+B_J^\intercal|=\kphEoneDefectGibbs$, the Bessel
	recurrence~\eqref{eq:bessel_recurrence} holds with residual
	$\kphEoneBessel$, and the Dirichlet-matrix entries match $(b,0,b/\beta)$;
	under the uniform measure the symmetric defect equals
	$\kphEoneDefectUniform$ at the $(\sin\theta,p)$ entry, confirming
	\eqref{eq:uniform_defect}. The mass matrix is positive definite,
	$\lambda_{\min}(M)=\kphEoneMineigM$.

	\emph{E2 (identification).} Fitting the three models of
	Section~\ref{sec:approximation} on training data and evaluating
	on held-out trajectories
	(Figure~\ref{fig:e2}), the structured convex fit reproduces the storage
	with error $\kphEtwoStorageFit$ and keeps its surrogate storage monotone
	along an unforced rollout, with $\kphEtwoIncStructured$ storage-increase
	steps against $\kphEtwoIncBaseline$ for the unconstrained baseline, at a
	comparable $\kphEtwoRollSteps$-step rollout error
	($\kphEtwoRollStructured$ against $\kphEtwoRollBaseline$).
    The separately output-consistent variant has rollout error
    $\kphEtwoRollOutput$ and $\kphEtwoIncOutput$ reported storage-increase
    step; it is not shown in the archived figure. Its saved diagnostics
    do not isolate that increase's solver or integration source.
    E2's archived output also lacks storage eigenvalues and conditioning,
    so it does not verify every positive-definiteness assumption used
    elsewhere in the theory.

	\emph{E3 (closed loop).} Figure~\ref{fig:e3} simulates the actual
    nonlinear plant with $b=0.1$ and $K_d=1$. Its plotted state is
    $z(t)=\Psi(x(t))=(\sin\theta(t),p(t))^\intercal$, without a shift;
    it is not an independently simulated surrogate trajectory. The
    eigenvector calculation gives $\kappa=\kphEthreeKappa$ and
    $\lambda=\kphEthreeLambda$ for the fitted $A_{\rm cl}$. The
    $60\times60$ grid on $[-2.5,2.5]^2$ gives the sampled maximum
    $\rho_{\rm grid}=\kphEthreeRhobar$ and the diagnostic radius
    $\kappa\rho_{\rm grid}/\lambda=\kphEthreeBound$. Neither a uniform
    residual enclosure nor invariance of this box was established by
    the original script. Moreover, the whole box has
    $\|\Psi\|\le\sqrt{7.25}<2.693$, so this radius is uninformative
    even as a diagnostic. All ten runs start below the plotted radius.
    The reported $\kphEthreeTermNorm$ is the maximum across runs of
    the mean of their final 40 sampled lift norms, not an asymptotic
    $\limsup$ estimate with a proof.

    The reproduced feedback is
    \begin{equation}
        u=-k_s\sin\theta-k_p p,\qquad
        k_s\simeq0.065267870796237,\quad k_p\simeq1.090865184886766.
        \label{eq:e3_recovered_controller}
    \end{equation}
    Its fitted storage is not a plant Lyapunov function, even locally:
    $P_{12}\simeq-0.005632269128536$, and at $p=0$,
    \[
    \frac{d}{dt}\left(\tfrac12\Psi^\intercal P\Psi\right)
    =-(1+k_s)P_{12}\sin^2\theta>0
    \]
    for arbitrarily small nonzero $\theta$. In contrast, the following
    physical argument certifies the same implemented controller.

    \begin{proposition}[Actual-plant shaped energy for the E3 controller]
        \label{prop:e3_physical}
        Let $a=1+k_s>0$, $d=0.1+k_p>0$, and consider
        \[
        \dot\theta=p,\qquad
        \dot p=-a\sin\theta-[d+\mu(1-\cos\theta)]p,
        \qquad \mu\ge0\text{ fixed}.
        \]
        Then $H_c=p^2/2+a(1-\cos\theta)$ satisfies
        \begin{equation}
        \dot H_c=-[d+\mu(1-\cos\theta)]p^2\le-dp^2.
        \label{eq:e3_shaped_energy}
        \end{equation}
        Every principal-well sublevel
        $\Omega_E=\{(\theta,p):|\theta|<\pi,\ H_c\le E\}$,
        $0<E<2a$, is compact and positively invariant, and all of its
        trajectories converge to $(0,0)$. The origin is asymptotically
        stable. On this set the implemented input obeys
        \begin{equation}
        |u|\le\sqrt{2E(k_s^2/a+k_p^2)}.
        \label{eq:e3_input_bound}
        \end{equation}
    \end{proposition}
    \begin{proof}
        Differentiating $H_c$ cancels the two $a\sin\theta\,p$ terms.
        The energy threshold keeps $|\theta|$ strictly below $\pi$
        and bounds $p$, giving a compact sublevel. The derivative is
        nonpositive, so the sublevel is invariant and trajectories are
        forward complete. In $\dot H_c=0$ one has $p=0$; to stay there
        requires $\sin\theta=0$, hence $\theta=0$ in this well.
        Positive definiteness of $H_c$ and LaSalle's principle yield
        stability and convergence. Weighted Cauchy--Schwarz and
        $\sin^2\theta\le2(1-\cos\theta)$ give
        $|u|^2\le(k_s^2/a+k_p^2)(a\sin^2\theta+p^2)\le
        2H_c(k_s^2/a+k_p^2)$.
    \end{proof}
    There is no largest closed sublevel strictly below $2a$. For the
    fixed comparison choice $E_*=a(1-\cos2.5)\simeq1.9187004243$,
    the region has $\max\|\Psi\|=\sqrt{2E_*}\simeq1.9589284950$
    (since $a>1$) and $|u|<2.141$. Only seven of the ten archived
    initial states lie in this particular region; eight lie below the
    separatrix. The regional proof does not cover the other initial
    states without a separate entry argument.

    This is classical energy shaping applied to the recovered learned
    controller, using the known plant equations. It neither follows
    from the fitted $P$ nor supplies a generic data-only certificate.
    Later validated comparisons on this same invariant region retain
    a serious successful competitor: a sector-aware quadratic
    certificate in $(\theta,p)$ proves a state decay rate $0.22890$,
    versus $0.05920$ for the tested physical-energy strictification,
    both with state prefactor at most $3$. The sector method uses the
    known plant and energy-based invariance with its own analysis
    nominal; it is not a norm-only test around the unchanged learned
    matrix. A uniformly enclosed absolute-residual radius remains
    about $17.395$, and an unstructured Euclidean proportional-error
    ball around that learned matrix large enough to contain the actual
    local residual cannot certify local asymptotic
    stability. These comparisons support retention of direction and
    sign, not superiority of physical energy over all residual methods.
    Appendix~\ref{app:e3_witnesses} supplies explicit baseline rate
    witnesses. Full matrices, continuous-region enclosures, and exact rational
    witness checks are provided in the repository's
    \texttt{numerics/results/certification/final\_validation/} package.
    The fixed-controller five-case damping sweep retains the same
    asymptotic region, but neither all residual metrics nor the
    certified exponential rates improve monotonically.

	\emph{E4 (state-dependent $J$).} On the unit ball the monomial dictionary
	of degree at most two gives exactly
	$\sum_i\partial_i\widehat{\mathbf{x}}_{ik}=0$, $B_J+B_J^\intercal=0$, and
	$B_J c_C=0$ symbolically; Monte Carlo estimates (Figure~\ref{fig:e4}) show
	the skew defect decaying at the fitted log-log slope $\kphEfourSlope$,
	reaching $\kphEfourDefectMaxN$ at $N=\kphEfourMaxN$.

	\begin{figure*}[t]
		\centering
		\includegraphics[width=0.92\textwidth]{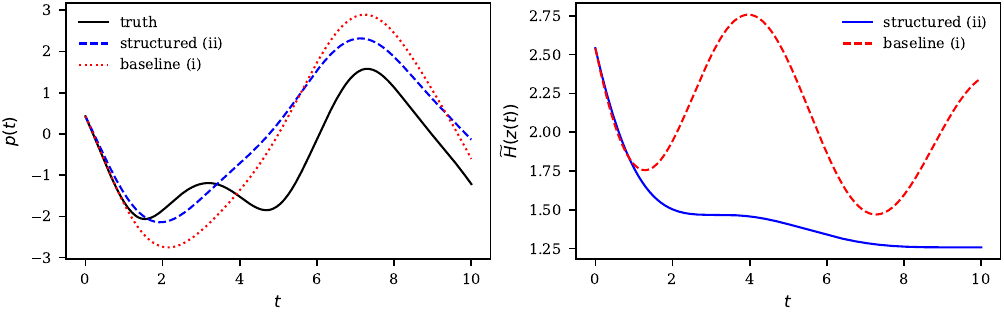}
		\caption{Experiment E2, damped pendulum. (a) Held-out rollout of the
		momentum $p(t)$: truth against the structured convex fit (ii) and the
		unconstrained baseline (i). (b) Surrogate storage
		$\widetilde{\mathcal{H}}(z(t))$ along an unforced rollout; the
		structured fit is monotone, matching the passivity inequality of
		Proposition~\ref{cor:passivity}, while the baseline is not.}
		\label{fig:e2}
	\end{figure*}

	\begin{figure*}[t]
		\centering
		\includegraphics[width=0.92\textwidth]{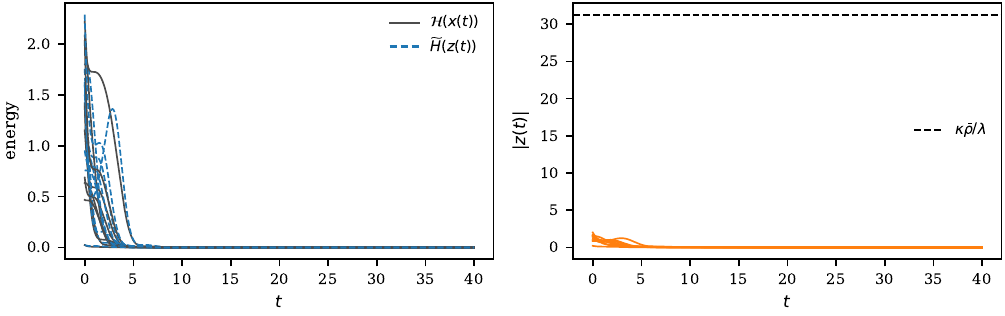}
		\caption{Archived E3 actual-plant closed loop. (a) Physical
        energy $\mathcal H(x(t))$ and fitted quadratic storage evaluated
        on the actual plant lift $z(t)=(\sin\theta(t),p(t))$.
        The fitted storage need not decrease. (b) Actual lift norms
        and the original grid-based radius (dashed). This sampled
        diagnostic is not a validated application of
        Theorem~\ref{thm:residual_bound}; all runs begin below it.
        Proposition~\ref{prop:e3_physical} gives a separate physical
        certificate for initial states in its specified sublevels.}
		\label{fig:e3}
	\end{figure*}

	\begin{figure*}[t]
		\centering
		\includegraphics[width=0.92\textwidth]{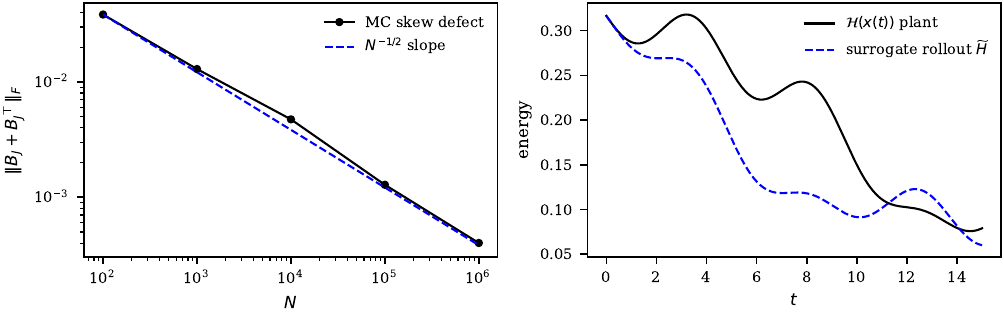}
		\caption{Experiment E4, rigid body. (a) Monte Carlo skew defect
		$\|B_J+B_J^\intercal\|_F$ against $N$ with an $N^{-1/2}$ reference,
		illustrating sampling error around the exact-skewness identity
		proved in Section~\ref{sec:example_rigidbody}. (b) Plant energy against the
		structured surrogate rollout.}
		\label{fig:e4}
	\end{figure*}

	\section{Conclusion}
	\label{sec:conclusion}
	
	This paper studied how the vector-field structure of
	port-Hamiltonian systems appears in the infinitesimal Koopman
	generator. For smooth observables the input-parametrized generator
	decomposes as in~\eqref{eq:gen_split}.
	The conservative component is formally skew only under the
	conservative-measure and boundary assumptions. Dissipative
	positivity is metric-bracket and Hamiltonian energy-balance
	positivity, not positivity of the first-order derivation
	$\mathcal{K}_R$ on arbitrary observables.
	
	For finite-dimensional approximation, the Galerkin weak-form
	projection must be stated with the mass matrix
	$M_{ij}=\langle\psi_i,\psi_j\rangle_\nu$. Under the conservative
	measure condition the weak matrix satisfies
	$B_J=-B_J^\intercal$, yielding the coefficient-space identity
	$A_J^\intercal M+MA_J=0$. The dissipative metric bracket yields
	the positive semidefinite Dirichlet matrix $D_R$, while the
	first-order dissipative derivation matrix
	$(C_R)_{ij}=\langle\psi_i,\mathcal{K}_R\psi_j\rangle_\nu$ is a
	different object and generally has no sign guarantee. Population
	identities are exact for the chosen Galerkin measure $\nu$;
	empirical matrices approximate those population quantities at
	finite sample size. Projection and closure residuals quantify the
	gap between exact lifted derivatives and a finite-dimensional
	surrogate.
	
	The fitted pH surrogate is passive by its imposed matrix
    constraints. Its damping-injection stability requires the stated
    detectability condition. These properties do not automatically
    transfer to the actual plant: $D_R\succeq0$ does not make the
    Galerkin drift dissipative in the mass-derived metric, and the
    fitted storage derivative contains a signed residual term.
    Proposition~\ref{prop:galerkin_compatibility} identifies the exact
    fixed-metric compatibility condition. The constant--Hamiltonian
    obstruction concerns that condition, not physical-energy
    representability in every dictionary or metric.

    The reproduced E3 controller admits a physical shaped-energy
    proof although its fitted quadratic storage increases at some
    plant states and its archived residual diagnostic is vacuous on
    the relevant region. A stronger sector-aware comparator also
    succeeds. The evidence supports using physical direction and sign
    in certification, with the available plant information stated
    explicitly. It does not establish a generic data-only plant
    guarantee, a new pendulum stabilizer, or superiority over all
    structured residual approaches. Larger application studies and
    uncertainty-aware verification remain necessary to establish a
    broader practical advantage.

    \appendices
    \section{Explicit baseline E3 decay witnesses}
    \label{app:e3_witnesses}
    This appendix supplies sufficient certificates on the invariant
    $\Omega_{E_*}$ of Proposition~\ref{prop:e3_physical} for $\mu=0$.
    The exported full-precision gains satisfy the rational enclosures
    \[
    1.065267\le a\le1.065268,\qquad
    1.190865\le d\le1.190866.
    \]
    All finite decimals used as witnesses below denote exact rationals.

    For the physical strictification take
    $W=H_c+\varepsilon\sin\theta\,p$ with
    $\varepsilon=0.670876796$ and $\gamma=0.1184$.
    Set $\eta=2\sqrt a\sin(\theta/2)$ and
    $t=\cos^2(\theta/2)\in[0.099,1]$ on the region. Direct differentiation
    writes $-\dot W-\gamma W$ as a quadratic form in $(\eta,p)$ with
    \[
    R_{11}=\varepsilon t-\gamma/2,\quad
    R_{22}=d+\varepsilon-2\varepsilon t-\gamma/2,\quad
    R_{12}^2=\frac{\varepsilon^2(d-\gamma)^2t}{4a}.
    \]
    Replacing $a$ by $1.065267$ in the negative term of the determinant
    gives a lower bound. That bound is separately concave in $t,d$;
    its four rectangle-corner values exceed $2.26\times10^{-5}$,
    and the diagonal entries are positive throughout the rectangle.
    These rational checks imply $\dot W\le-\gamma W$ everywhere
    on the region, not only at sampled plant states.

    Put $r=\varepsilon/\sqrt a<2/3$. Then
    $(1-r)H_c\le W\le(1+r)H_c$. The elementary bounds on
    $|\theta|\le2.5$ give
    \[
    \frac{a(1-\cos2.5)}{2.5^2}\|x\|^2\le H_c(x)
    \le\frac a2\|x\|^2.
    \]
    The lower coefficient is less than $1/2$, as required for the
    momentum term. Since $\cos2.5<-0.8$, the squared state prefactor
    is less than $5(2.5)^2/[2(1-\cos2.5)]<625/72<9$.
    Consequently $\|x(t)\|\le3e^{-0.0592t}\|x(0)\|$.

    For the sector comparator write the same plant as
    $\dot x=F(k,b)x$, where
    \[
    F(k,b)=\begin{bmatrix}0&1\\-k&-b\end{bmatrix},\quad
    k=a\,\frac{\sin\theta}{\theta},\quad k(0)=a,\quad b=d.
    \]
    Monotonicity of $\sin\theta/\theta$ on $[0,2.5]$ gives
    $k\in[0.255,1.065268]$; the lower bound follows, for example,
    from the alternating sine series through degree $23$.
    Choose
    \[
    Q=\begin{bmatrix}
    4.0995902693&3.6973495816\\
    3.6973495816&6.2087874765
    \end{bmatrix},\qquad \beta=0.2289.
    \]
    Exact rational leading-minor checks give $I\prec Q\prec9I$.
    At the four vertices of the above $k$ interval and
    $b\in[1.190865,1.190866]$, the matrix
    $-F^\intercal Q-QF-2\beta Q$ has positive leading minors and
    determinant greater than $0.00264$. Affine dependence on $k,b$
    proves the inequality throughout the rectangle. Together with
    the independent energy-based invariance of $\Omega_{E_*}$,
    this yields $\|x(t)\|\le3e^{-0.2289t}\|x(0)\|$.
    The vertex-LMI construction is standard
    \cite{boyd1994lmi}; it is a strong model-informed comparator.
    Exact checks for these readable rational witnesses are included
    with the revision audit. Neither rate is claimed optimal over
    all possible Lyapunov functions.

\section*{Acknowledgment}
	The author thanks Prof.\ Dr.\ Matthias Ehrhardt for identifying bibliographic
	errors in an earlier version of this manuscript.

	\bibliographystyle{IEEEtran}
	\bibliography{references_v20}

\end{document}

%% file: numbers.tex
\newcommand{\kphEoneDefectGibbs}{1.1\times10^{-15}}
\newcommand{\kphEoneDefectUniform}{-0.50}
\newcommand{\kphEoneBessel}{0}
\newcommand{\kphEoneMineigM}{0.446}
\newcommand{\kphEtwoStorageFit}{1.5\times10^{-6}}
\newcommand{\kphEtwoRollBaseline}{2.60}
\newcommand{\kphEtwoRollStructured}{2.40}
\newcommand{\kphEtwoRollOutput}{2.36}
\newcommand{\kphEtwoIncBaseline}{108}
\newcommand{\kphEtwoIncStructured}{0}
\newcommand{\kphEtwoIncOutput}{1}

\newcommand{\kphEtwoRollSteps}{200}
\newcommand{\kphEthreeKappa}{4.94}
\newcommand{\kphEthreeLambda}{0.425}
\newcommand{\kphEthreeRhobar}{2.69}
\newcommand{\kphEthreeBound}{31.3}
\newcommand{\kphEthreeTermNorm}{5.9\times10^{-10}}
\newcommand{\kphEfourSlope}{-0.50}
\newcommand{\kphEfourDefectMaxN}{4.0\times10^{-4}}
\newcommand{\kphEfourMaxN}{10^{6}}

%% file: references_v20.bib
@book{van2006port,
  author    = {van der Schaft, Arjan},
  title     = {$L_2$-Gain and Passivity Techniques in
               Nonlinear Control},
  edition   = {3rd},
  publisher = {Springer},
  address   = {Cham},
  year      = {2017},
  doi       = {10.1007/978-3-319-49992-5}
}

@article{schaft2014port,
  author    = {van der Schaft, Arjan and Jeltsema, Dimitri},
  title     = {Port-{H}amiltonian Systems Theory:
               An Introductory Overview},
  journal   = {Foundations and Trends in Systems and Control},
  volume    = {1},
  number    = {2--3},
  pages     = {173--378},
  year      = {2014},
  doi       = {10.1561/2600000002}
}

@article{ortega2001interconnection,
  author    = {Ortega, Romeo and van der Schaft, Arjan and
               Maschke, Bernhard and Escobar, Gerardo},
  title     = {Interconnection and Damping Assignment
               Passivity-Based Control of Port-Controlled
               {H}amiltonian Systems},
  journal   = {Automatica},
  volume    = {38},
  number    = {4},
  pages     = {585--596},
  year      = {2002},
  doi       = {10.1016/S0005-1098(01)00278-3}
}

@article{ortega2001energyshaping,
  author    = {Ortega, Romeo and van der Schaft, Arjan and
               Mareels, Iven and Maschke, Bernhard},
  title     = {Putting energy back in control},
  journal   = {{IEEE} Control Systems Magazine},
  volume    = {21},
  number    = {2},
  pages     = {18--33},
  year      = {2001},
  doi       = {10.1109/37.915398}
}

@article{koopman1931hamiltonian,
  author    = {Koopman, Bernard O.},
  title     = {Hamiltonian Systems and Transformation
               in {H}ilbert Space},
  journal   = {Proceedings of the National Academy of
               Sciences},
  volume    = {17},
  number    = {5},
  pages     = {315--318},
  year      = {1931},
  doi       = {10.1073/pnas.17.5.315}
}

@article{mezic2005spectral,
  author    = {Mezi\'c, Igor},
  title     = {Spectral Properties of Dynamical Systems,
               Model Reduction and Decompositions},
  journal   = {Nonlinear Dynamics},
  volume    = {41},
  number    = {1--3},
  pages     = {309--325},
  year      = {2005},
  doi       = {10.1007/s11071-005-2824-x}
}

@article{mezic2013analysis,
  author    = {Mezi\'c, Igor},
  title     = {Analysis of Fluid Flows via Spectral
               Properties of the {K}oopman Operator},
  journal   = {Annual Review of Fluid Mechanics},
  volume    = {45},
  pages     = {357--378},
  year      = {2013},
  doi       = {10.1146/annurev-fluid-011212-140652}
}

@article{brunton2022modern,
  author    = {Brunton, Steven L. and Budi{\v{s}}i{\'c}, Marko
               and Kaiser, Eurika and Kutz, J. Nathan},
  title     = {Modern {K}oopman Theory for Dynamical Systems},
  journal   = {{SIAM} Review},
  volume    = {64},
  number    = {2},
  pages     = {229--340},
  year      = {2022},
  doi       = {10.1137/21M1401243}
}

@book{eisner2015operator,
  author    = {Eisner, Tanja and Farkas, B{\'a}lint and
               Haase, Markus and Nagel, Rainer},
  title     = {Operator Theoretic Aspects of Ergodic Theory},
  publisher = {Springer},
  address   = {Cham},
  year      = {2015},
  doi       = {10.1007/978-3-319-16898-2}
}

@book{lasota1994chaos,
  author    = {Lasota, Andrzej and Mackey, Michael C.},
  title     = {Chaos, Fractals, and Noise: Stochastic
               Aspects of Dynamics},
  edition   = {2nd},
  publisher = {Springer},
  address   = {New York},
  year      = {1994},
  doi       = {10.1007/978-1-4612-4286-4}
}

@article{klus2020generator,
  author    = {Klus, Stefan and N{\"u}ske, Feliks and
               Peitz, Sebastian and Niemann, Jan-Hendrik and
               Clementi, Cecilia and Sch{\"u}tte, Christof},
  title     = {Data-driven Approximation of the {K}oopman
               Generator: Model Reduction, System
               Identification, and Control},
  journal   = {Physica D: Nonlinear Phenomena},
  volume    = {406},
  year      = {2020},
  note      = {Art. no. 132416, doi: 10.1016/j.physd.2020.132416},
  doi       = {10.1016/j.physd.2020.132416}
}

@article{willems1972dissipative,
  author    = {Willems, Jan C.},
  title     = {Dissipative Dynamical Systems {P}art {I}:
               General Theory},
  journal   = {Archive for Rational Mechanics and Analysis},
  volume    = {45},
  number    = {5},
  pages     = {321--351},
  year      = {1972},
  doi       = {10.1007/BF00276493}
}

@article{korda2018linear,
  author    = {Korda, Milan and Mezi\'c, Igor},
  title     = {Linear Predictors for Nonlinear Dynamical
               Systems: {K}oopman Operator Meets Model
               Predictive Control},
  journal   = {Automatica},
  volume    = {93},
  pages     = {149--160},
  year      = {2018},
  doi       = {10.1016/j.automatica.2018.03.046}
}

@article{dahdah2022systemnorm,
  author    = {Dahdah, Steven and Forbes, James R.},
  title     = {System Norm Regularization Methods for
               {K}oopman Operator Approximation},
  journal   = {Proceedings of the Royal Society A},
  volume    = {478},
  number    = {2265},
  year      = {2022},
  note      = {Art. no. 20220162, doi: 10.1098/rspa.2022.0162},
  doi       = {10.1098/rspa.2022.0162}
}

@article{yildiz2024,
  author    = {Y{\i}ld{\i}z, S{\"u}leyman and Goyal, Pawan and
               Bendokat, Thomas and Benner, Peter},
  title     = {Data-Driven Identification of Quadratic
               Representations for Nonlinear {H}amiltonian
               Systems Using Weakly Symplectic Liftings},
  journal   = {Journal of Machine Learning for Modeling and
               Computing},
  volume    = {5},
  number    = {2},
  pages     = {45--71},
  year      = {2024},
  doi       = {10.1615/JMachLearnModelComput.2024052810}
}

@article{morandin2023,
  author    = {Morandin, Riccardo and Nicodemus, Jonas
               and Unger, Benjamin},
  title     = {Port-{H}amiltonian Dynamic Mode Decomposition},
  journal   = {{SIAM} Journal on Scientific Computing},
  volume    = {45},
  number    = {4},
  pages     = {A1690--A1710},
  year      = {2023},
  note      = {doi: 10.1137/22M149329X},
  doi       = {10.1137/22M149329X}
}

@article{cherifi2022,
  author    = {Cherifi, Karim and Goyal, Pawan
               and Benner, Peter},
  title     = {A Non-Intrusive Method to Inferring Linear
               Port-{H}amiltonian Realizations Using
               Time-Domain Data},
  journal   = {Electronic Transactions on Numerical Analysis},
  volume    = {56},
  pages     = {102--116},
  year      = {2022},
  note      = {Special Issue SciML, doi: 10.1553/etna\_vol56s102},
  doi       = {10.1553/etna_vol56s102}
}

@article{hara2020dissipativity,
  author    = {Hara, Keita and Inoue, Masaki and Sebe, Noboru},
  title     = {Learning {K}oopman Operator under
               Dissipativity Constraints},
  journal   = {{IFAC}-{P}apersOnLine},
  volume    = {53},
  number    = {2},
  pages     = {1169--1174},
  year      = {2020},
  note      = {doi: 10.1016/j.ifacol.2020.12.1327},
  doi       = {10.1016/j.ifacol.2020.12.1327}
}

@book{khalil2002nonlinear,
  author    = {Khalil, Hassan K.},
  title     = {Nonlinear Systems},
  edition   = {3rd},
  publisher = {Prentice Hall},
  year      = {2002}
}

@article{kaufman1984,
  author    = {Kaufman, Allan N.},
  title     = {Dissipative {H}amiltonian Systems: A Unifying Principle},
  journal   = {Physics Letters A},
  volume    = {100},
  number    = {8},
  pages     = {419--422},
  year      = {1984},
  doi       = {10.1016/0375-9601(84)90634-0}
}

@article{morrison1986,
  author    = {Morrison, Philip J.},
  title     = {A Paradigm for Joined {H}amiltonian and
               Dissipative Systems},
  journal   = {Physica D: Nonlinear Phenomena},
  volume    = {18},
  number    = {1--3},
  pages     = {410--419},
  year      = {1986},
  doi       = {10.1016/0167-2789(86)90209-5}
}

@article{grmela1997,
  author    = {Grmela, Miroslav and {\"O}ttinger, Hans Christian},
  title     = {Dynamics and Thermodynamics of Complex Fluids.
               {I}. Development of a General Formalism},
  journal   = {Physical Review E},
  volume    = {56},
  number    = {6},
  pages     = {6620--6632},
  year      = {1997},
  doi       = {10.1103/PhysRevE.56.6620}
}

@article{grmela1997b,
  author    = {{\"O}ttinger, Hans Christian and Grmela, Miroslav},
  title     = {Dynamics and Thermodynamics of Complex Fluids.
               {II}. Illustrations of a General Formalism},
  journal   = {Physical Review E},
  volume    = {56},
  number    = {6},
  pages     = {6633--6655},
  year      = {1997},
  doi       = {10.1103/PhysRevE.56.6633}
}

@book{ottinger2005,
  author    = {{\"O}ttinger, Hans Christian},
  title     = {Beyond Equilibrium Thermodynamics},
  publisher = {Wiley},
  address   = {Hoboken, NJ},
  year      = {2005},
  doi       = {10.1002/0471727903}
}

@article{bloch1996,
  author    = {Bloch, Anthony and Krishnaprasad, P. S. and
               Marsden, Jerrold E. and Ratiu, Tudor S.},
  title     = {The {E}uler--{P}oincar\'e Equations and
               Double Bracket Dissipation},
  journal   = {Communications in Mathematical Physics},
  volume    = {175},
  number    = {1},
  pages     = {1--42},
  year      = {1996},
  doi       = {10.1007/BF02101622}
}

@article{morrison2024,
  author    = {Morrison, Philip J. and Updike, Michael H.},
  title     = {Inclusive Curvaturelike Framework for Describing
               Dissipation: Metriplectic 4-Bracket Dynamics},
  journal   = {Physical Review E},
  volume    = {109},
  number    = {4},
  pages     = {045202},
  year      = {2024},
  doi       = {10.1103/PhysRevE.109.045202}
}

@article{weinstein1997,
  author    = {Weinstein, Alan},
  title     = {The Modular Automorphism Group of a
               {P}oisson Manifold},
  journal   = {Journal of Geometry and Physics},
  volume    = {23},
  number    = {3--4},
  pages     = {379--394},
  year      = {1997},
  doi       = {10.1016/S0393-0440(97)80011-3}
}

@book{mauroy2020,
  editor    = {Mauroy, Alexandre and Mezi\'c, Igor and
               Susuki, Yoshihiko},
  title     = {The {K}oopman Operator in Systems and Control:
               Concepts, Methodologies, and Applications},
  publisher = {Springer},
  address   = {Cham},
  year      = {2020},
  series    = {Lecture Notes in Control and Information Sciences},
  volume    = {484},
  doi       = {10.1007/978-3-030-35713-9}
}

@inproceedings{surana2016,
  author    = {Surana, Amit},
  title     = {Koopman Operator Based Observer Synthesis for
               Control-Affine Nonlinear Systems},
  booktitle = {Proc. 55th {IEEE} Conf. Decision and Control (CDC)},
  pages     = {6492--6499},
  year      = {2016},
  doi       = {10.1109/CDC.2016.7799268}
}

@inproceedings{goswami2017,
  author    = {Goswami, Debdipta and Paley, Derek A.},
  title     = {Global Bilinearization and Controllability of
               Control-Affine Nonlinear Systems: A {K}oopman
               Spectral Approach},
  booktitle = {Proc. 56th {IEEE} Conf. Decision and Control (CDC)},
  pages     = {6107--6112},
  year      = {2017},
  doi       = {10.1109/CDC.2017.8264582}
}

@article{peitz2020,
  author    = {Peitz, Sebastian and Otto, Samuel E. and
               Rowley, Clarence W.},
  title     = {Data-Driven Model Predictive Control Using
               Interpolated {K}oopman Generators},
  journal   = {{SIAM} Journal on Applied Dynamical Systems},
  volume    = {19},
  number    = {3},
  pages     = {2162--2193},
  year      = {2020},
  doi       = {10.1137/20M1325678}
}

@article{gonzalez1996,
  author    = {Gonzalez, Oscar},
  title     = {Time Integration and Discrete
               {H}amiltonian Systems},
  journal   = {Journal of Nonlinear Science},
  volume    = {6},
  number    = {5},
  pages     = {449--467},
  year      = {1996},
  doi       = {10.1007/BF02440162}
}

@article{mclachlan1999,
  author    = {McLachlan, Robert I. and Quispel, G. R. W. and
               Robidoux, Nicolas},
  title     = {Geometric Integration Using Discrete Gradients},
  journal   = {Philosophical Transactions of the Royal Society A},
  volume    = {357},
  number    = {1754},
  pages     = {1021--1045},
  year      = {1999},
  doi       = {10.1098/rsta.1999.0363}
}

@article{messenger2021,
  author    = {Messenger, Daniel A. and Bortz, David M.},
  title     = {Weak {SINDy}: {G}alerkin-Based Data-Driven
               Model Selection},
  journal   = {Multiscale Modeling \& Simulation},
  volume    = {19},
  number    = {3},
  pages     = {1474--1497},
  year      = {2021},
  doi       = {10.1137/20M1343166}
}

@misc{mamakoukas2020stable,
  author = {Mamakoukas, Giorgos and Abraham, Ian and Murphey, Todd D.},
  title = {Learning Stable Models for Prediction and Control},
  year = {2022},
  howpublished = {arXiv:2005.04291v2},
  note = {Version 2, March 24, 2022},
  url = {https://arxiv.org/abs/2005.04291v2}
}

@article{straesser2026safedmd,
  author = {Str{"a}{\ss}er, Robin and Schaller, Manuel and Worthmann, Karl
            and Berberich, Julian and Allg{\"o}wer, Frank},
  title = {{SafEDMD}: A {Koopman}-based data-driven controller design
           framework for nonlinear dynamical systems},
  journal = {Automatica},
  volume = {185},
  pages = {112732},
  year = {2026},
  doi = {10.1016/j.automatica.2025.112732}
}

@book{boyd1994lmi,
  author = {Boyd, Stephen and El Ghaoui, Laurent and Feron, Eric
            and Balakrishnan, Venkataramanan},
  title = {Linear Matrix Inequalities in System and Control Theory},
  publisher = {Society for Industrial and Applied Mathematics},
  year = {1994},
  volume = {15},
  series = {Studies in Applied Mathematics},
  doi = {10.1137/1.9781611970777},
  url = {https://web.stanford.edu/~boyd/lmibook/}
}

@misc{xu2025dissipativity,
  author = {Xu, Yuezhu and Sivaranjani, S. and Gupta, Vijay},
  title = {Learning Neural {Koopman} Operators with Dissipativity Guarantees},
  year = {2025},
  howpublished = {arXiv:2509.07294},
  url = {https://arxiv.org/abs/2509.07294}
}

@article{mayne2000constrained,
  author = {Mayne, D. Q. and Rawlings, J. B. and Rao, C. V.
            and Scokaert, P. O. M.},
  title = {Constrained model predictive control: Stability and optimality},
  journal = {Automatica},
  volume = {36},
  number = {6},
  pages = {789--814},
  year = {2000},
  doi = {10.1016/S0005-1098(99)00214-9}
}
